\documentclass[12pt,dvipsnames]{article}\usepackage[]{graphicx}\usepackage[]{color}
\makeatletter
\def\maxwidth{ %
  \ifdim\Gin@nat@width>\linewidth
    \linewidth
  \else
    \Gin@nat@width
  \fi
}
\makeatother

\definecolor{fgcolor}{rgb}{0.345, 0.345, 0.345}

\usepackage{framed}
\makeatletter
 {\par\unskip\endMakeFramed%
 \at@end@of@kframe}
\makeatother

\definecolor{shadecolor}{rgb}{.97, .97, .97}
\definecolor{messagecolor}{rgb}{0, 0, 0}
\definecolor{warningcolor}{rgb}{1, 0, 1}
\definecolor{errorcolor}{rgb}{1, 0, 0}
\newenvironment{knitrout}{}{} % an empty environment to be redefined in TeX

\usepackage{alltt}

\usepackage{amsmath}
\usepackage{graphicx,psfrag,epsf}
\usepackage{enumerate}
\usepackage{natbib}
\usepackage{url} % not crucial - just used below for the URL 

\usepackage{amssymb}% to use \mathbb for blackboard bold
\usepackage{amsthm}% to set up theorems
\usepackage{bm}% for bold math
\usepackage{float}% to place figures with [H]
\usepackage[bbgreekl]{mathbbol}% for blackboard bold greek letters
\usepackage{tikz}% to use tikzfigure
\usepackage[normalem]{ulem}% to use \sout in feedback commands
\usepackage{xcolor}% to use colors in feedback commands
\usepackage{xspace}% to use \xspace inside \newcommand to check for need of trailing space

\usepackage{booktabs} % for the supplement; used for \toprule, \midrule, and \bottomrule rather than \hline for horizontal line tags in tables produced by R pacakge xtable

\newcommand\msTitle{Panel data analysis via mechanistic models}

\newcommand\PanelPOMP{PanelPOMP}

\newcommand\SuppSecLikAvg{S2\xspace}
\newcommand\SuppTabAlgPars{S-1\xspace}

\newcommand\Np{J}   % Number of particles
\newcommand\nrep{r} % Replicate index (in likelihood evaluations)
\newcommand\Nrep{R} % Number of replicates (in likelihood evaluations)
\newcommand\unit{u} % Unit index
\newcommand\Unit{U} % Number of panel units
\newcommand\Nmif{M} % Number of iterated filtering steps
\newcommand\nmif{m} % Mif index
\newcommand\lambdaMCAP{\lambda_{\mathrm{loess}}}

\newcommand\prob{\mathbb{P}}
\newcommand\E{\mathbb{E}}
\newcommand\var{\mathrm{Var}}

\newcommand\real{\mathbb{R}}
\newcommand\normal{\mathrm{Normal}}

\newcommand\lik{\ell}
\newcommand\Xspace{{\mathbb X}}
\newcommand\Yspace{{\mathbb Y}}
\newcommand\Thetaspace{\mathbb{\Theta}}
\newcommand\dimSpecific{{\mathrm{dim}(\Psi)}}
\newcommand\dimShared{{\mathrm{dim}(\Phi)}}

\newcommand\Thetadim{{{\mathrm{dim}}(\Thetaspace)}}
\newcommand\contactsRate{\Lambda}

\newcommand\given{{\, | \,}}
\newcommand\giventh{{\,;\,}}
\newcommand\mycolon{{\hspace{0.6mm}:\hspace{0.6mm}}}

\newcommand\listA{A}
\newcommand\listR{R}

\newcommand\Cii{C_1}
\newcommand\Ciii{C_2}
\newcommand\Civ{C_3}
\newcommand\Cv{C_4}

\newcommand\mclik[2]{L_{#1}^{(#2)}}
\newcommand\mcUnit{\bar L}
\newcommand\likA{\widetilde L}
\newcommand\likB{\widehat L}

\newcommand\revision[1]{#1\xspace}

\newtheorem{theorem}{Theorem}
\newcommand{\conditionList}{ %\setlength{\leftmargin}{0.2cm}
\setlength{\topsep}{1mm} \setlength{\itemsep}{0.5mm} \setlength{\parsep}{0cm}}

\newcommand\bigO{\mathcal{O}}
\newcommand\negBelowProfile{\vspace{-8mm}}
\newcommand{\blind}{1}

\usepackage{fullpage}% to set all 4 margins to 1 inch and specify page style
\IfFileExists{upquote.sty}{\usepackage{upquote}}{}
\begin{document}
\def\spacingset#1{\renewcommand{\baselinestretch} {#1}\small\normalsize} 
\spacingset{1}
\if1\blind
{
\title{\bf \msTitle}
\author{Carles Bret\'o and Edward L. Ionides \\
  Department of Statistics, University of Michigan\\
  and \\
  Aaron A. King\thanks{This research was supported by National Science Foundation grant DMS-1308919 and National Institutes of Health grants U54-GM111274, U01-GM110712 and R01-AI101155.
}\hspace{.2cm}\\
  Department of Ecology and Evolutionary Biology, University of Michigan}
\maketitle
} \fi
\if0\blind
{
  \bigskip
  \bigskip
  \bigskip
  \begin{center}
    {\LARGE\bf \msTitle}
\end{center}
  \medskip
} \fi
\bigskip
\begin{abstract}
%The text of your abstract. 200 or fewer words.
\noindent 
Panel data, also known as longitudinal data, consist of a collection of time series. Each time series, which could itself be multivariate, comprises a sequence of measurements taken on a distinct unit. Mechanistic modeling involves writing down scientifically motivated equations describing the collection of dynamic systems giving rise to the observations on each unit. A defining characteristic of panel systems is that the dynamic interaction between units should be negligible. Panel models therefore consist of a collection of independent stochastic processes, generally linked through shared parameters while also having unit-specific parameters. To give the scientist flexibility in model specification, we are motivated to develop a framework for inference on panel data permitting the consideration of arbitrary nonlinear, partially observed panel models. We build on iterated filtering techniques that provide likelihood-based inference on nonlinear partially observed Markov process models for time series data. Our methodology depends on the latent Markov process only through simulation; this plug-and-play property ensures applicability to a large class of models. We demonstrate our methodology on a toy example and two epidemiological case studies.
%% begin{shared-comm.tex}
\newcommand\abstractComplexity{issues arising due to the combination of model complexity and dataset size.}%
%% end{shared-comm.tex}
We address inferential and computational \revision{\abstractComplexity}
\end{abstract}
\noindent%
\vfill
\newpage

%\spacingset{1.45} % DON'T change the spacing!
\spacingset{1.7} % this gives 26 lines per page using the fullpage format
%% END{JASA's preamble template}

% begin{R setup}

% Rmd header

% end{R setup}

\section{Introduction} \label{sec:intro}
%\linenumbers
Analyzing time series data on a collection of related units provides opportunities to study aspects of dynamic systems---their replicability, or dependence on properties of the units---that cannot be revealed from measurements on a single unit.
The units might be individual humans or animals, in an observational or experimental study. 
The units might also be spatial locations, giving a panel representation of spatiotemporal data.
As a consequence of advances in data collection, scientists investigating dynamic systems have growing capabilities to obtain measurements of increasing length on increasingly many units.
Statistical investigation of such data, known as panel data analysis, is therefore playing a growing role in the scientific process.

Mechanistic modeling of a dynamic system involves writing down equations describing the evolution of the system through time.
Time series analysis using mechanistic models involves determining whether the model provides an adequate description of the system, and if so, identifying plausible values for unknown parameters \citep{breto09}.
Stochasticity, nonlinearity and noisy incomplete observations are characteristic features of many systems in the biological and social sciences \citep{bjornstad01,dobson14}.
Monte Carlo inference approaches have been developed that are effective for general classes of models with these properties. 
Such methods include iterated filtering \citep{ionides06,ionides15}, particle Markov chain Monte Carlo \citep{andrieu10} and synthetic likelihood \citep{wood10}.
All these inference algorithms obtain their general applicability by enjoying the plug-and-play property, that is, they interface with the dynamic model only through simulation \citep{breto09,he10}.
However, these methodologies do not address the particular structure of panel models and the high-dimensional nature of panel data. 
Therefore, new methodology is required to analyze panel data when there is a need to consider models outside the linear, Gaussian paradigm.
We proceed by building on the iterated filtering approach of \citet{ionides15}, deriving a panel iterated filtering likelihood maximization algorithm.
The panel iterated filtering algorithm, an associated convergence theorem, and a software implementation equipped with an appropriate domain-specific modeling language, all extend the existing theory and practice of iterated filtering.

Across the broad applications of nonlinear partially observed stochastic dynamic models for time series analysis \citep{douc14} one can anticipate many situations where multiple time series are available and give rise to the structure of panel data.
In particular, panel data on dynamic systems arises in pharmacokinetics \citep{donnet13}, molecular biology \citep{chen16}, infectious disease transmission \citep{cauchemez04,yang10,yang12}, and microeconomics \citep{heiss08,bartolucci12,mesters14}.
Our methodology differs from those employed by these authors in that it provides plug-and-play likelihood-based inference applicable to general nonlinear, non-Gaussian models.
This scope of applicability also sets our goals apart from the extensive panel methodology literature building on a linear regression framework \citep[e.g.,][]{hsiao14}.

In Section~\ref{sec:meth} we present a basic panel iterated filtering algorithm, and in Section~\ref{sec:pif} we prove its convergence under appropriate regularity conditions. 
An issue arising for large panel datasets is scalability of statistical methodology, and we develop three techniques to address this issue in Sections~\ref{sec:MCAP},~\ref{sec:marginal} and~\ref{sec:replication}.
These scaling techniques are illustrated on a toy example in Section~\ref{sec:gompertz}.
Two scientifically motivated examples follow: modeling the transmission of polio in Section~\ref{sec:polio}, and dynamic variation in human sexual contact rates in Section~\ref{sec:hiv}.
In Section~\ref{sec:dis}, we conclude by discussing relationships to other approaches and indicating some extensions.

\section{Inference methodology: panel iterated filtering (PIF)} \label{sec:meth}

Units of a panel are labeled $\{1,2,\dots,\Unit\}$, which we write as $1\mycolon\Unit$.
The $N_\unit$ measurements collected on unit $\unit$ are written as $y^*_{\unit,1:N_{\unit}} = \{y^*_{\unit,1},\dots,y^*_{\unit,N_\unit}\}$ where $y^*_{\unit,n}$ is collected at time $t_{\unit,n}$ with $t_{\unit,1}<t_{\unit,2}<\dots<t_{\unit,N_\unit}$.
These data are considered fixed and modeled as a realization of an observable stochastic process $Y_{\unit,1:N_\unit}$.
This observable process is constructed to be dependent on a latent Markov process $\{X_\unit(t),t_{\unit,0}\le t\le t_{\unit,N_\unit}\}$, for some $t_{\unit,0}\le t_{\unit,1}$.
Further requiring that $\{X_\unit(t)\}$ and $\{Y_{\unit,i},i\neq n\}$ are independent of $Y_{\unit,n}$ given $X_\unit(t_{\unit,n})$, for each $n\in 1\mycolon N_{\unit}$, completes the structure required for a partially observed Markov process (POMP) model for unit~$\unit$ \revision{\citep{ionides11,king16}}.
If all units are modeled as independent, the model is called a {\PanelPOMP}.
Although we can treat time as either continuous or discrete, our attention will focus on the latent process at the observation times, so we write $X_{\unit,n}=X_\unit(t_{\unit,n})$.
We suppose that $X_{\unit,n}$ and $Y_{\unit,n}$ take values in arbitrary spaces $\Xspace_{\unit}$ and $\Yspace_{\unit}$ respectively.
We suppose that the joint density of $X_{\unit,0:N_\unit}$ and $Y_{\unit,1:N_\unit}$ exists, with respect to some suitable measure, and is written as $f_{X_{\unit,0:N_\unit}Y_{\unit,1:N_\unit}}(x_{\unit,0:N_\unit},y_{\unit,1:N_\unit}\giventh\theta)$,
with dependence on an unknown parameter $\theta\in\Thetaspace\subset\real^\Thetadim$.
Each component of the vector $\theta$ may affect one, several or all units.
This framework encompasses fixed effects (discussed in Section~\ref{sec:marginal}) and random effects (discussed in Section~\ref{sec:dis}).
The transition density $f_{X_{\unit,n}|X_{\unit,n-1}}(x_{\unit,n}\given x_{\unit,n-1}\giventh\theta)$ and measurement density $f_{Y_{\unit,n}|X_{\unit,n}}(y_{\unit,n}\given x_{\unit,n}\giventh\theta)$ are permitted to depend arbitrarily on $u$ and $n$, allowing non-stationary models and the inclusion of covariate time series (illustrated in Section~\ref{sec:polio}).
The framework also includes continuous-time dynamic models (illustrated in Section~\ref{sec:hiv}) and discrete-time dynamic models (illustrated in Sections~\ref{sec:gompertz} and~\ref{sec:polio}), for which $X_{\unit,0:N_\unit}$ is specified directly without ever defining $\{X_\unit(t),t_{\unit,0}\le t\le t_{\unit,N_\unit}\}$.

The marginal density of $Y_{\unit,1:N_\unit}$ at $y_{\unit,1:N_\unit}$ is $f_{Y_{\unit,1:N_\unit}}(y_{\unit,1:N_\unit}\giventh\theta)$ and the likelihood function for unit $\unit$ is
$\lik_{\unit}(\theta) = f_{Y_{\unit,1:N_\unit}}(y^*_{\unit,1:N_\unit}\giventh\theta)$.
The likelihood for the entire panel is 
$\lik(\theta) = \prod_{\unit=1}^{\Unit} \lik_{\unit}(\theta)$,
and any solution $\hat\theta=\arg\max\lik(\theta)$ is a maximum likelihood estimate (MLE).

\newpage

\spacingset{1.4}

\newcommand\mystretch{\rule[-2mm]{0mm}{5mm} }
\newcommand\asp{\hspace{4mm}}

\noindent\begin{tabular}{l}
\hline
{\bf Algorithm~{PIF}. Panel iterated filtering}\rule[-1.5mm]{0mm}{6mm}\\
\hline
{\bf input:}\rule[-1.5mm]{0mm}{6mm} \\
Simulator of initial density, $f_{X_{\unit,0}}(x_{\unit,0}\giventh \theta)$ for $\unit$ in $1{\mycolon}\Unit$ \\
Simulator of transition density, $f_{X_{\unit,n}|X_{\unit,n-1}}(x_{\unit,n}\given x_{\unit,n-1}\giventh\theta)$ for $\unit$ in $1{\mycolon}\Unit$, $n$ in $1{\mycolon}N_\unit$ \\
Evaluator of measurement density, $f_{Y_{\unit,n}|X_{\unit,n}}(y_{\unit,n}\given x_{\unit,n}\giventh\theta)$ for $\unit$ in $1{\mycolon}\Unit$, $n$ in $1{\mycolon}N_\unit$ \\
Data, $y_{\unit,n}^*$ for $\unit$ in $1{\mycolon}\Unit$ and $n$ in $1\mycolon N_\unit$ \\
Number of iterations, $\Nmif$ \\
Number of particles, $J$ \\
Starting parameter swarm, $\Theta^0_{j}$ for $j$ in $1\mycolon J$\\
Simulator of perturbation density, $h_{\unit,n}(\theta\given\varphi\giventh\sigma)$ for $\unit$ in $1{\mycolon}\Unit$, $n$ in $0{\mycolon}N_\unit$\\
Perturbation sequence, $\sigma_{m}$ for $\nmif$ in $1\mycolon M$ \\
{\bf output:}\rule[-1.5mm]{0mm}{6mm} \\
Final parameter swarm, $\Theta^M_j$ for $j$ in $1\mycolon J$\\
\hline
For $\nmif$ in $1\mycolon M$\rule[0mm]{0mm}{5mm}\\
\asp Set $\Theta^m_{0,j}=\Theta^{m-1}_j$ for  $j$ in $1\mycolon J$\mystretch\\
\asp For $\unit$ in $1\mycolon \Unit$\\
\asp \asp Set $\Theta^{F,m}_{\unit,0,j}\sim h_{\unit,0}\big(\theta\given\Theta^{m}_{\unit-1,j}\giventh\sigma_{m}\big)$ for  $j$ in $1\mycolon J$\mystretch\\
\asp \asp Set $X_{\unit,0,j}^{F,m}\sim f_{X_{\unit,0}}(x_{\unit,0} \giventh \Theta^{F,m}_{\unit,0,j})$ for $j$ in $1\mycolon J$\mystretch\\
\asp \asp For $n$ in $1\mycolon N_{\unit}$\\
\asp \asp \asp $\Theta^{P,m}_{\unit,n,j}\sim h_{\unit,n}(\theta\given\Theta^{F,m}_{\unit,n-1,j},\sigma_{m})$ for $j$ in $1\mycolon J$\mystretch\\
\asp \asp \asp $X_{\unit,n,j}^{P,m}\sim f_{X_{\unit,n}|X_{\unit,n-1}}(x_{\unit,n} \given X^{F,m}_{\unit,n-1,j}\giventh\Theta^{P,m}_{\unit,n,j})$ for $j$ in $1\mycolon J$  \mystretch\\
\asp \asp \asp $w_{\unit,n,j}^m = f_{Y_{\unit,n}|X_{\unit,n}}(y^*_{\unit,n}\given X_{\unit,n,j}^{P,m} \giventh \Theta^{P,m}_{\unit,n,j})$ for $j$ in $1\mycolon J$  \mystretch\\
\asp \asp \asp Draw $k_{1:J}$ with $\prob(k_j=i)=  w_{\unit,n,i}^m\Big/\sum_{v=1}^J w_{\unit,n,v}^m$ for $i,j$ in $1\mycolon J$\\
\asp \asp \asp $\Theta^{F,m}_{\unit,n,j}=\Theta^{P,m}_{\unit,n,k_j}$ and $X^{F,m}_{\unit,n,j}=X^{P,m}_{\unit,n,k_j}$ for $j$ in $1\mycolon J$   \mystretch\\
\asp \asp End For \\ %%% end n loop
\asp \asp Set $\Theta^{m}_{\unit,j}=\Theta^{F,m}_{\unit,N_{\unit},j}$ for $j$ in $1\mycolon J$\\
\asp End For \\ %%% end i loop
\asp Set $\Theta^{m}_j=\Theta^m_{\Unit,j}$ for $j$ in $1\mycolon J$\\
End For\\ %% end m loop
\hline
\end{tabular}

\clearpage
\spacingset{1.7}

The PIF algorithm, represented by the pseudocode above, is an adaptation of the IF2 algorithm \citep{ionides15} to {\PanelPOMP} models.
%% begin{shared-comm.tex}
\newcommand\describePIF{Like previous iterated filtering algorithms \citep{ionides06,ionides15}, PIF explores the space of unknown parameters by stochastically perturbing them and applying sequential Monte Carlo to filter the data seeking for parameter values that are concordant with the data.
Perturbations are successively diminished over repeated filtering iterations, leading to convergence to a maximum likelihood estimate.
This general approach of iterated filtering needs some adaptation in order for it to be useful for PanelPOMP models.}%
%% end{shared-comm.tex}
\revision{\describePIF}

%% begin{shared-comm.tex}
\newcommand\OJMNU{The number of computations required for PIF has order $\bigO(JMN{\Unit})$, where $N$ is the mean of $\{N_1,\dots,N_\Unit\}$ and $J$ and $M$ are the number of particles and iterations, defined in the pseudocode.}%
%% end{shared-comm.tex}
\revision{\OJMNU}
The pseudocode specifies unique labels for each quantity constructed to clarify the logical structure of the algorithm, and a literal implementation of this pseudocode therefore requires storing $\bigO(JMN{\Unit})$ particles.
Each particle contains a perturbed parameter vector and so has size $\bigO(\Unit)$ if $\Thetadim$ is $\bigO(\Unit)$, leading to a total storage requirement of $\bigO(JMN{\Unit}^2)$.
However, we only need to store the value of the latent process particles, $X_{\unit,n,1:J}^{P,m}$ and $X_{\unit,n,1:J}^{F,m}$, and perturbed parameter particles, $\Theta_{\unit,n,1:J}^{P,m}$ and $\Theta_{\unit,n,1:J}^{F,m}$, for the current unit, time point and PIF iteration.
Taking advantage of this memory over-writing opportunity leads to a storage requirement that is $\bigO(J{\Unit})$. 

The theoretical justification of PIF is based on the observation that a {\PanelPOMP} model can be represented as a time-inhomogeneous POMP model.
Algorithms for {\PanelPOMP}s, and their theoretical support, can therefore be derived from previous approaches for POMPs.
Here, we use a representation concatenating the time series for each unit, corresponding to a latent POMP process
\begin{equation}\label{eq:concatenation}
X(t) = X_\unit \Big( t_{\unit,0} + \left(t-T^{\,\mathrm{cum}}_{\unit-1}\right) \Big) \mbox{ for }
T^{\,\mathrm{cum}}_{\unit-1} \le t \le T^{\,\mathrm{cum}}_{\unit}-1,
\end{equation}
where $T^{\,\mathrm{cum}}_{\unit}$ is a cumulative latent POMP process time for all panel units up to unit $\unit$, given by
\begin{equation}\label{eq:Tcum}
T^{\,\mathrm{cum}}_{\unit}=\unit+\sum_{k=1}^{\unit} \big(t_{k,N_k}-t_{k,0}\big)
\end{equation}
and $T^{\,\mathrm{cum}}_0=0$.
We leave $X(t)$ undefined for $T^{\,\mathrm{cum}}_{\unit}-1 < t < T^{\,\mathrm{cum}}_{\unit+1}$ to provide a formal separation  between the latent processes for each unit. 
In \eqref{eq:Tcum} we have set the value of this time separation to one, though any positive number would suffice and the exact value is irrelevant on the discrete timescale consisting of the sequence of observation times.
%% begin{shared-comm.tex}
\newcommand\describePIFconcatenation{In the language of data manipulation, our representation converts wide panel data into a tall format \citep{wickham14}.
As we show subsequently, a POMP representation using a tall format preserves the theoretical justification for iterated filtering, while also taking advantage of favorable sequential Monte Carlo (SMC) stability properties for long time series.
Conversely, a wide format POMP representation risks encountering the curse of dimensionality for SMC \citep{bengtsson08}.}%
%% end{shared-comm.tex}
\revision{\describePIFconcatenation}

%% begin{shared-comm.tex}
\newcommand\differenceWithPrevious{%
This POMP representation of a PanelPOMP model is one of three noted by \citet{romero-severson15} and discussed further in the supplement (Section~S1).
\citet{romero-severson15} used a different algorithm---their approach was convenient to code and sufficient for their example but its computational feasibility quickly breaks down as the length of each panel time series increases so it is infeasible in situations such as our polio example in Section~6.}%
%% end{shared-comm.tex}
\revision{\differenceWithPrevious}

\section{Convergence of PIF} \label{sec:pif}

PIF investigates the parameter space using a particle swarm $\Theta^m_{1:J} = \{\Theta^m_j, \mbox{$j \in 1{\mycolon}J$}\}$.
With a sufficiently large number $J$ of particles, each iteration $\nmif$ of PIF approximates a Bayes map that selects a particle $j$ with probability proportional to the value of the likelihood function at $\Theta^m_j$.
Heuristically, repeated application of the Bayes map favors particles with high likelihood and should lead to convergence of the particle swarm to a neighborhood of the MLE.
We state such a convergence theorem, followed by the technical assumptions we use to prove it.

Our theorem combines Theorems~1 and~2 of \citet{ionides15} in the context of the POMP representation of a {\PanelPOMP} model in \eqref{eq:concatenation}.
%% begin{shared-comm.tex}
\newcommand\mathInnovation{Their Theorem~1 proved the existence of a limit distribution for an iterated perturbed Bayes map by taking advantage of its linearity under the Hilbert projective metric. In addition, they showed that sequential Monte Carlo can provide a uniform approximation of this limit distribution. Their Theorem~2 bounded excursion probabilities under this iterated perturbed Bayes map to derive sufficient conditions for the limit distribution to concentrate around the MLE.}%
%% end{shared-comm.tex}
\revision{\mathInnovation}
Here, we combine these two theorems into a simpler statement.

\begin{theorem}\label{thm:pif}
Let $\Theta^M_{1:J}$ be the output of PIF, with fixed perturbations $\sigma_m=\delta$.
Suppose regularity conditions \listA\ref{A1}--\listA\ref{A6} below.
For all $\epsilon>0$, there exists $\delta$, $M_0$ and $C$ such that, for all $M\ge M_0$ and all $j \in 1{\mycolon}J$,
\begin{equation}
\prob\Big[ \big |\Theta^M_j-\hat\theta\big| >\epsilon \Big] < \epsilon + \frac{C}{\sqrt{J}}.
\end{equation}
\end{theorem}
To discuss the regularity conditions, we need to set up some more notation.
We write $Y=(Y_{1,1:N_1},\dots,Y_{\Unit,1:N_\Unit})$ and consequently we write $y^*$ for a vector of the entire panel dataset.
The likelihood function is $\lik(\theta) = f_{Y}(y^*\giventh\theta)$ and we suppose the following regularity condition:
\newcounter{A}
\newcounter{savedA}
\begin{list}{(\listA\arabic{A})}{\usecounter{A}\conditionList}
\item\label{A1} There is a unique MLE, and $\lik(\theta)$ is continuous in a neighborhood of this MLE.
\setcounter{savedA}{\theA}
\end{list}
To allow us to talk about parameter perturbations, we define a perturbed parameter space,
$$\breve\Thetaspace= \Thetaspace^{N_1+1} \times \Thetaspace^{N_2+1} \times\dots\times \Thetaspace^{N_\Unit+1},$$
for which we write $\breve\theta\in\breve\Thetaspace$ as
\begin{equation}\label{breveTheta}
\breve\theta=\big(\theta_{1,0},\theta_{1,1},\dots,\theta_{1,N_1},\theta_{2,0},\dots,\theta_{2,N_2},\dots,\theta_{\Unit,N_\Unit}\big).
\end{equation}
For compatibility with the POMP representation of a PanelPOMP in \eqref{eq:concatenation}, perturbed parameters for each time point and each unit are concatenated in \eqref{breveTheta} with $\theta_{\unit,n}$ being a perturbed parameter for the $n$th observation on unit $\unit$.
On the perturbed parameter space $\breve\Thetaspace$, the extended likelihood function is defined as
\begin{eqnarray}
\breve\lik(\breve\theta)
&=& \nonumber
\prod_{\unit=1}^{N_\Unit} \int\dots\int dx_{\unit,0} \dots dx_{\unit,N_\Unit} \bigg\{
f_{X_{\unit,0}}(x_{\unit,0}\giventh\breve\theta_{\unit,0})
\\
&&\hspace{2cm}\prod_{n=1}^{N_\unit}
f_{X_{\unit,n}|X_{\unit,n-1}}(x_{\unit,n}\given x_{\unit,n-1}\giventh\breve\theta_{\unit,n})f_{Y_{\unit,n}|X_{\unit,n}}(y^*_{\unit,n}\given x_{\unit,n}\giventh\breve\theta_{\unit,n})
\bigg\}.
\end{eqnarray}
We suppose that the extended likelihood has a Lipschitz continuity property:
\begin{list}{(\listA\arabic{A})}{\usecounter{A}\conditionList}
\setcounter{A}{\thesavedA}
\item\label{A2}
Set $\breve N = \sum_{\unit=1}^\Unit (N_\unit+1)$, so that $\breve\Thetaspace=\Thetaspace^{\breve N}$.
Write $\breve\theta_n$ for the $n$th of the $\breve N$ terms in (\ref{breveTheta}), so that $\breve\theta=\breve\theta_{1:\breve N}$.
There is a $\Cii$ such that
\begin{equation}
\big|\breve\lik(\breve\theta)-\lik(\theta_{1,0})\big|
<
\Cii \sum_{n=2}^{\breve N} \big|\breve \theta_n -\breve \theta_{n-1}\big|.
\end{equation}
\setcounter{savedA}{\theA}
\end{list}
We also assume a uniformly positive measurement density:
\begin{list}{(\listA\arabic{A})}{\usecounter{A}\conditionList}
\setcounter{A}{\thesavedA}
\item\label{A3}
There are constants $\Ciii$ and $\Civ$ such that
$$0<\Ciii< f_{Y_{\unit,n}|X_{\unit,n}}(y_{\unit,n}^*\given x_{\unit,n}\giventh\theta)<\Civ<\infty,$$
for all $\unit\in 1{\mycolon}\Unit$, $n\in 1{\mycolon}N_\unit$, $x_{\unit,n}\in\Xspace_{\unit}$ and $\theta\in\Thetaspace$.
\setcounter{savedA}{\theA}
\end{list}
This condition will usually require $\Thetaspace$  and $\Xspace_\unit$ to be compact, for all $\unit\in 1\mycolon\Unit$.
Compactness of $\Thetaspace$ is satisfied if there is some limit to the scientifically plausible values of each parameter.
Compactness of $\Xspace_\unit$ may not be satisfied in practice, but much previous theory for SMC has used this strong requirement \citep[e.g.,][]{delmoral01,legland04}.
The remaining conditions concern the perturbation transition density, $h_{\unit,n}(\theta\given\varphi\giventh\sigma)$.
We suppose that $h_{\unit,n}(\theta\given\varphi\giventh\sigma)$ has bounded support on a normalized scale via the following condition:
\begin{list}{(\listA\arabic{A})}{\usecounter{A}\conditionList}
\setcounter{A}{\thesavedA}
\item\label{A4}
There is a $\Cv$ such that $h_{\unit,n}(\theta\given\varphi\giventh\sigma)=0$ when $|\theta-\varphi|<\Cv \sigma$ for all $\unit\in 1{\mycolon}\Unit$, $n\in 1{\mycolon}N_\unit$ and $\sigma$.
\setcounter{savedA}{\theA}
\end{list}
We also require some regularity of an appropriately rescaled limit of the Markov chain resulting from iterating the perturbation process.
Define $\{\breve\Theta_m, m\ge 0\}$ to be a Markov chain taking values in $\Thetaspace$ with $\breve\Theta_0$ drawn uniformly from the starting particles, $\Theta^0_{1:J}$, and transition density given by
\begin{eqnarray}
f_{\breve\Theta_m|\breve\Theta_{m-1}}(\theta_{\Unit,N_\Unit}\given\varphi\giventh\sigma)
&=& \nonumber
\int
\bigg\{
h_{1,0}(\theta_{1,0}\given\varphi\giventh\sigma)
\prod_{\unit=2}^{\Unit} h_{\unit,0}(\theta_{\unit,0}\given\theta_{\unit-1,N_{\unit-1}}\giventh\sigma)\, d\theta_{\unit-1,N_{\unit-1}}
\\
&&
\hspace{3cm}
\prod_{\unit=1}^{\Unit}\prod_{n=1}^{N_\unit} h_{\unit,n}(\theta_{\unit,n}\given\theta_{\unit,n-1}\giventh\sigma)\, d\theta_{\unit,n-1}
\bigg\}.
\end{eqnarray}
Thus, $\{\breve\Theta_m\}$ represents a random-walk-like process corresponding to combining the parameter perturbations of all units and all time points for one iteration of PIF.
Now, let $\{W_\sigma(t),t\ge 0\}$ be a right-continuous, piecewise constant process taking values in $\Thetaspace$ defined at its points of discontinuity by
\begin{equation}
W_\sigma(k\sigma^2)= \breve\Theta_{k}.
\end{equation}
If $h_{\unit,n}(\theta\given\varphi\giventh\sigma)$ were a scale family of additive perturbations, then $\{\breve\Theta_m\}$ would be a random walk that scales to a Brownian diffusion.
When $\Thetaspace$ has a boundary, $\{\breve\Theta_m\}$ cannot be exactly a random walk, but similar diffusion limits can apply \citep{bossy04}.
We require that $h_{\unit,n}(\theta\given\varphi\giventh\sigma)$ is chosen to be sufficiently regular to have such a diffusion limit, via the following assumptions:
\begin{list}{(\listA\arabic{A})}{\usecounter{A}\conditionList}
\setcounter{A}{\thesavedA}
\item \label{A5} $\{W_\sigma(t),0\le t\le 1\}$ converges weakly as $\sigma\to 0$ to a diffusion $\{W(t),0\le t\le 1\}$, in the space of right-continuous functions with left limits equipped with the uniform convergence topology.
For any open set $A\subset \Thetaspace$ with positive Lebesgue measure and $\epsilon>0$, there is a $\delta(A,\epsilon)>0$ such that $\prob\big[W(t)\in A \mbox{ for all } \epsilon\le t\le 1\given W(0)\big]>\delta$.
\item\label{A6} For some $t_0(\sigma)$ and $\sigma_0>0$, $W_\sigma(t)$ has a positive density on $\Thetaspace$, uniformly over the distribution of $W(0)$ for all $t>t_0$ and $\sigma<\sigma_0$.
\setcounter{savedA}{\theA}
\end{list}
\begin{proof}[Proof of Theorem~\ref{thm:pif}]
PIF is exactly the IF2 algorithm of \citet{ionides15} applied to the POMP representation of a {\PanelPOMP} model in equation~\eqref{eq:concatenation}.
\listA\ref{A1} is condition B3 of \citet{ionides15} together with a simplifying additional assumption that a unique MLE exists.
\listA\ref{A2} is a re-writing of their B6.
\listA\ref{A3} and \listA\ref{A4} are essentially their B4 and B5.
\listA\ref{A5} and \listA\ref{A6} match their B1 and B2, respectively.
Thus, we have established the conditions used for Theorems~1 and~2 of \citet{ionides15} for this POMP representation of a {\PanelPOMP} model.
The statement of our Theorem~\ref{thm:pif} follows directly from these two previous results.
\end{proof}

The perturbation density $h_{\unit,n}(\theta\given\varphi\giventh\sigma)$ has usually been chosen to be Gaussian in implementations of the IF2 algorithm and its predecessor, the IF1 algorithm \citep{ionides06}.
The use of Gaussian perturbations requires the user to reparameterize boundaries in the parameter space.
%% begin{shared-comm.tex}
\newcommand\transformations{This may involve taking a logarithmic transformation of positive parameters and a logistic transform of interval-valued parameters (for example, see Tables~S3 and~S4).}%
\newcommand\truncation{In order to satisfy assumption A4, Gaussian perturbations must be truncated. Since the Gaussian distribution has short tails, ignoring truncation is practically equivalent to truncation at a large multiple of the standard deviation.}%
\newcommand\stein{Unlike an alternative theoretical framework using Stein's lemma to approximate derivatives via perturbed parameters and SMC \citep{doucet15,dai16} our justification for PIF does not require the choice of Gaussian perturbations.}%
%% end{shared-comm.tex}
\revision{\transformations} \revision{\truncation} \revision{\stein}

\section{Scalable methodology for large panels} \label{sec:scale}

Theorem~\ref{thm:pif} provides an asymptotic Monte Carlo convergence guarantee for a dataset of fixed size as the Monte Carlo effort increases.
In practice, reaching this \revision{Monte Carlo}  asymptotic regime becomes increasingly difficult as the panel dataset grows, whether the number of units becomes large, or there are many observations per unit, or both.
In this section, we consider some techniques that become important when using PIF for big datasets.
We demonstrate the methodology on a toy model in Section~\ref{sec:gompertz}.
Subsequently, we demonstrate data analysis for mechanistic panel models using PIF via two examples, one investigating disease transmission of polio and another investigating dynamics of human sexual behavior.

\subsection{Monte Carlo adjusted profile (MCAP) confidence intervals}\label{sec:MCAP}

PIF provides a Monte Carlo approach to maximizing the likelihood function for a {\PanelPOMP} model.
It is based on an SMC algorithm that can also provide an unbiased estimate of the maximized likelihood.
Monte Carlo methods to evaluate and maximize the likelihood function provide a foundation for constructing confidence intervals via profile likelihood.
When computational resources are sufficient to make Monte Carlo error small, its role in statistical inference may be negligible.
With large datasets and complex models, we cannot ignore Monte Carlo error so instead we quantify it and draw statistical inferences that properly account for it.
We use the Monte Carlo adjusted profile (MCAP) methodology of \citet{ionides17} which fits a smooth curve through Monte Carlo evaluations of points on a profile log likelihood.
MCAP then obtains a confidence interval using a cutoff on this estimated profile that is enlarged to give proper coverage despite Monte Carlo uncertainty in its construction.
Monte Carlo variability in maximizing and evaluating the likelihood both lead to expected under-estimation of the maximized log likelihood.
Despite such bias the MCAP methodology remains valid as long as this \emph{likelihood shortfall} is slowly varying as a function of the profiled parameter
\citep{ionides17}.
Our toy example in Section~\ref{sec:gompertz} demonstrates this phenomenon.
%% begin{shared-comm.tex}
\newcommand\MCAP{For our subsequent examples, we applied the Monte Carlo adjusted profile (MCAP) methodology described by \citet{ionides17} and demonstrated in several recent scientific analyses \citep{smith17,ranjeva17,ranjeva18,pons-salort18}.
We used the R implementation of MCAP from \citet{ionides17} with an algorithmic smoothing parameter $\lambdaMCAP=0.9$ determining the fraction of profile points used to construct the neighborhoods for locally weighted quadratic regression smoothing.}%
%% end{shared-comm.tex}
\revision{\MCAP}

\subsection{Unit-specific parameters and marginal maximization}\label{sec:marginal}

The parameter space for a {\PanelPOMP} model may be structured into unit-specific and shared parameters.
To do this, we introduce a decomposition $\Thetaspace=\Phi\times\Psi^\Unit$ and $\theta=(\phi,\psi_1,\psi_2,\dots,\psi_\Unit)$ with the joint distribution of $X_{\unit,0:N_\unit}$ and $Y_{\unit,1:N_\unit}$, for each unit $\unit$, depending only on the shared parameter $\phi\in\Phi\subset\real^\dimShared$ and the unit-specific parameter $\psi_u\in\Psi\subset\real^\dimSpecific$.
%% begin{shared-comm.tex}
\newcommand\unitSpecific{The general PanelPOMP specification does not insist on the existence of unit-specific parameters.
Correspondingly, the PIF algorithm does not require this structure and it does not play a role in the general theory.
However, when it exists, we can use this additional structure to advantage within the general framework of PIF.
One consequence of the presence of unit-specific parameters arises in a natural structure for the PIF perturbation densities: when PIF is filtering through panel unit $u$, only unit-specific parameters corresponding to unit $u$ need to be perturbed.
Another consequence is the existence of a block structure to the parameter space that can be exploited, as follows.}%
%% end{shared-comm.tex}
\revision{\unitSpecific}

When $\Unit$ is large, $\Thetadim=\dimShared + \Unit\dimSpecific$ also becomes large. 
For a fixed value of $\phi$, the marginal likelihoods of $\psi_{1:\Unit}$ can be maximized separately, due to independence between units. 
Formally, application of iterated filtering to each of these marginal optimizations is a special case of the PIF algorithm with $\Unit=1$.
Thus, the convergence theory for PIF gives us freedom to alternate marginal optimization steps with joint optimization over $\Thetaspace$, following a block coordinate ascent approach.
In practice, we demonstrate a simple two-step algorithm which first attempts to optimize over $\Thetaspace$ and then refines the resulting estimates of the unit-specific parameters by marginal searches for each unit.
Figure~\ref{fig:gompertz} of Section~\ref{sec:gompertz} shows that this leads to considerable Monte Carlo variance reduction on an analytically tractable example.

\subsection{Using replications for likelihood maximization and evaluation}\label{sec:replication}

Monte Carlo replication, with differing random number generator seed values, is a basic tool for obtaining and assessing Monte Carlo approximations to a maximum likelihood estimate and its corresponding maximized likelihood.
Replication is trivially parallelizable, so provides a simple strategy to take advantage of large numbers of computer processors. 
Repeated searches, from wide-ranging starting values, provide a practical assessment of the success of global maximization.
When many Monte Carlo searches have found a comparable maximized likelihood, and no searches have surpassed it, we have some confidence that the likelihood surface has been adequately investigated.

PIF requires an additional calculation to evaluate the likelihood at the proposed maximum. 
The PIF algorithm produces an estimate of the likelihood for the perturbed model, and if the perturbations are small this may provide a useful approximation to the likelihood, however re-evaluation with the unperturbed model is appropriate for likelihood-based inference. 
Making $R$ replicated Monte Carlo evaluations of the likelihood gives rise to estimates $\{\lik^{(\nrep)}_{\unit}, \nrep \in 1{\mycolon}\Nrep, \unit\in 1{\mycolon}\Unit\}$ for each replication and unit. 
One possible way to combine these is an estimate $\tilde \lik = \frac{1}{\Nrep}\sum_{\nrep=1}^{\Nrep} \prod_{\unit=1}^{\Unit} \lik^{(\nrep)}_{\unit}$. 
When computed via SMC, this estimate is unbiased \citep[Theorem 7.4.2 on page 239 in][]{delmoral04}.
However, we use an alternative unbiased estimate, $\hat \lik = \prod_{u=1}^U\frac{1}{R}\sum_{r=1}^R  \lik^{(r)}_u$, which has lower variance (derived in Section~\SuppSecLikAvg).

%%gggggggggggggggggggg

\section{A toy example: the panel Gompertz model} \label{sec:gompertz}

We consider a {\PanelPOMP} constructed as a stochastic version of the discrete-time Gompertz model for biological population growth \citep{winsor32}. 
We suppose that the density, $X_{\unit,n+1}$, of a population $\unit$ at time $n+1$ depends on the density, $X_{\unit,n}$, at time $n$ according to
\begin{equation}
\label{eq:gompertz1}
X_{\unit,n+1}=\kappa_\unit^{1-e^{-r_\unit}}\,X_{\unit,n}^{e^{-r_\unit}}\,\varepsilon_{\unit,n}.
\end{equation}
In \eqref{eq:gompertz1}, $\kappa_\unit$ is the carrying capacity of population $\unit$, $r_\unit$ is a positive parameter, and $\{\varepsilon_{\unit,n},\unit\in 1\mycolon\Unit,n\in 1\mycolon N_\unit\}$ are independent and identically-distributed lognormal random variables with $\log\varepsilon_{\unit,n}\sim\normal(0,\sigma_{\mathrm{G},\unit}^2)$.
We model the population density to be observed with errors in measurement that are lognormally distributed:
\begin{equation}
%\label{eq:gompertz-obs}
\nonumber \log{Y_{\unit,n}}\;\sim\;\normal\left(\log{X_{\unit,n}},\tau_\unit^2\right).
\end{equation}
The Gompertz model is a convenient toy nonlinear non-Gaussian model since it has a logarithmic transformation to a linear Gaussian process and therefore the exact likelihood is computable by the Kalman filter \citep{king16}.
%% begin{shared-comm.tex}
\newcommand\nonnegligibleMC{The simulation experiment is designed to have non-negligible Monte Carlo error in order to test the effectiveness of the combined PIF and MCAP algorithms in this situation.}%
%% end{shared-comm.tex}
\revision{\nonnegligibleMC}
As discussed in Section~\ref{sec:MCAP}, we expect Monte Carlo estimates of profile log likelihood functions to fall below the actual (usually unknown) value. 
This is in part because imperfect maximization can only reduce the maximized likelihood, and in part a consequence of Jensen's inequality applied to the likelihood evaluation: the unbiased SMC likelihood evaluation has a negative bias on estimation of the log likelihood.
However, this bias produces a vertical shift in the estimated profile that may (and, in this example, does) have negligible effect on the resulting confidence interval.

%^% Define functions used in all examples

%*% run gompertz example 

\begin{figure}
\begin{knitrout}
\definecolor{shadecolor}{rgb}{1, 1, 1}\color{fgcolor}

{\centering \includegraphics[width=3.5in]{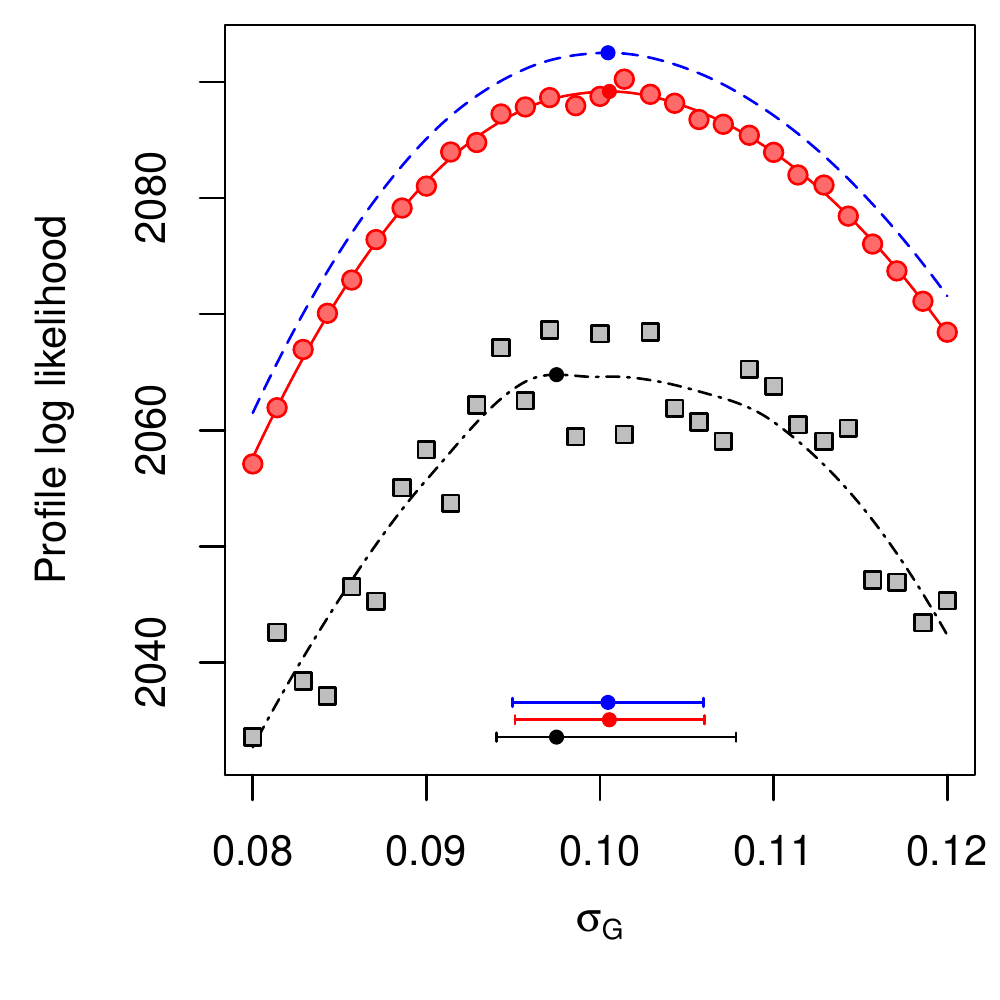} 

}

\end{knitrout}
\negBelowProfile
\caption{Profile log likelihood of $\sigma_{\mathrm{G}}$ for a panel of size $\Unit=50$ for the Gompertz model.
Blue line (dashes): exact profile.
Red points and line (circles): profile computed with PIF, including marginal maximization for unit-specific parameters.
Black points and line (squares): profile computed with PIF using only joint maximization.
The horizontal bars show $95$\% MCAP confidence intervals with a small filled circle marking the MLE obtained with algorithmic parameters in table~\SuppTabAlgPars in the supplement.}\label{fig:gompertz}
\end{figure}

For our experiment, we used $N_\unit=100$ simulated observations for each of $\Unit=50$ panel units.
For each $\unit\in1\mycolon\Unit$, we fixed $\kappa_\unit=1$ and $X_{\unit,0}=1$.
We set $\sigma_{\mathrm{G},\unit}=\sigma_{\mathrm{G}}=0.1$ and $r_\unit=r=0.1$.
We estimated the shared parameters $\sigma_{\mathrm{G}}$ and $r$.
We also estimated unit-specific parameters $\tau_{1:\Unit}$ with true values set to $\tau_\unit=0.1$.
We profiled over the shared parameter $\sigma_{\mathrm{G}}$, maximizing with respect to $r$ and the 50 unit-specific parameters $\tau_{1:\Unit}$.
In Figure~\ref{fig:gompertz}, the estimated profile using the marginal step has a log likelihood shortfall of only approximately $3.4 / 51 = 0.07$ log units per parameter.
By contrast, maximization using only the joint step has a shortfall of $28.1 / 51 = 0.6$ per parameter and substantially greater Monte Carlo variability.
This greater variability leads to a larger Monte Carlo adjusted profile cut-off than the asymptotic value of 1.92, and therefore typically produces a wider 95\% confidence interval \citep{ionides17}.

\section{Polio: state-level pre-vaccination incidence in USA}\label{sec:polio}

%*% run polio example 

The study of ecological and epidemiological systems poses challenges involving nonlinear mechanistic modeling of partially observed processes \citep{bjornstad01}.
Here, we illustrate this class of models and data, in the context of panel data analysis, by analyzing state-level historic polio incidence data in USA.
Although introduction of a pathogen into a host community requires contact between communities, the vast majority of infectious disease transmission events have infector and infectee within the same community \citep{bjornstad02}.
Therefore, for the purpose of understanding the dynamics of infectious diseases within communities, one may choose to model a collection of communities as independent conditional on a pathogen immigration process.
For example, fitting a panel model to epidemic data on a collection of geographical regions could permit statistical identification of dynamic mechanisms that cannot readily be detected by the data available in any one region.
Further, differences between regions (in terms of size, climate, and other demographic or geographic factors) may lead to varying disease dynamics that can challenge and inform a panel model.

The massive efforts of the global polio eradication initiative have brought polio from a major global disease to the brink of extinction \citep{patel16}.
Finishing this task is proving hard, and an improved understanding of polio ecology might assist.
\citet{martinez-bakker15} investigated polio dynamics by fitting a mechanistic model to pre-vaccination era data from the USA consisting of monthly reports of acute paralysis from May 1932 through January 1953.
Reports are available for the 48 contiguous US states and Washington D.C., so $\Unit=49$, and henceforth we refer to these units as states.
A sample of the time series in this panel is plotted in the supplement.
\citet{martinez-bakker15} fitted their model separately to each state which, in panel terminology, amounted to a decision to make all parameters unit-specific.
Some parameters, such as duration of infection, might be well modeled as shared between all units.
Other parameters, such as the model for seasonality of disease transmission, should intuitively be slowly varying geographically.
\citet{martinez-bakker15} did not have access to panel inference methodology, and so here we reconsider their model and data and investigate what happens when some parameters become shared between units.

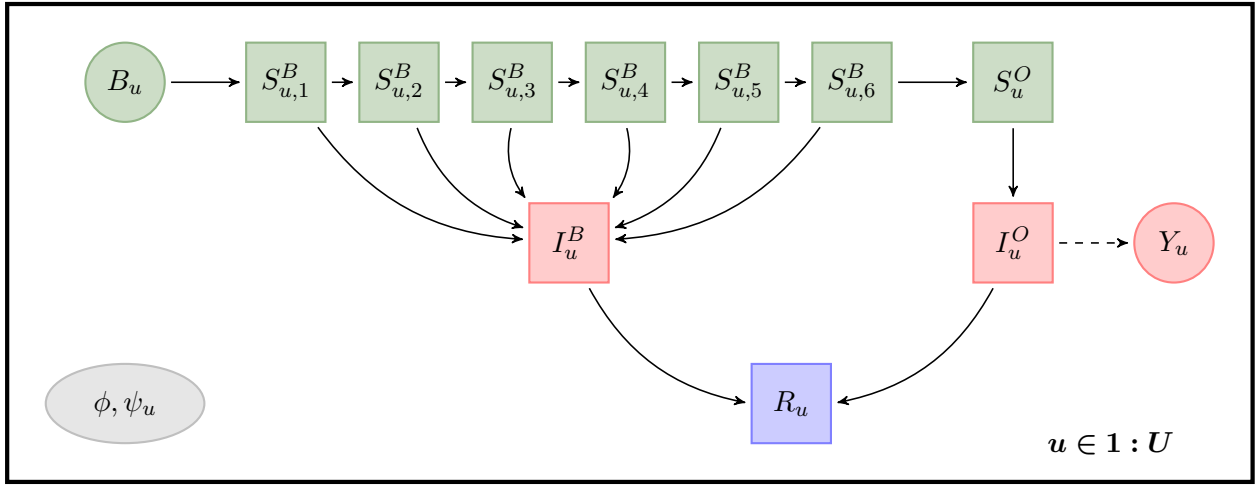
\begin{figure}
\centering
% DIAGRAM
\usetikzlibrary{arrows}% to have fancier "stealth" arrow tips
\usetikzlibrary{backgrounds}% to treat as background the curly braces between which the main tikzfigure is placed
\usetikzlibrary{calc}% to center one node exactly between two others
\usetikzlibrary{fit}% to fit rectangle around compartment model
\usetikzlibrary{positioning}% to specify tikz's node relative position with 'below=of,' [above right=0.7cm and 4cm of A],' etc
\usetikzlibrary{shapes}% to use ellipse
% BEGIN colors, sizes, & distances
\newcommand\col[2]{#1#2}% to combine colors (matching compartment type) with different intensities (for line vs fill)
% specify SIR colors
\newcommand\Scol{OliveGreen}%\newcommand\Scol{magenta}%
\newcommand\Icol{red}%\newcommand\Icol{cyan}%
\newcommand\Rcol{blue}%\newcommand\Rcol{red}%
\newcommand\Paramscol{gray}%\newcommand\Rcol{red}%
\newcommand\Bcol{\Scol}
% specify draw/fill intensities
\newcommand\drawIntens{!50}
\newcommand\fillIntens{!20}
\newcommand\textIntens{!99}
% speficy size
\newcommand\compartSize{30pt}% to set tikz's 'minimum size' option common to all compartments
\newcommand\compartDist{2*\compartSize}% so that tikz's 'node distance' option scales with 'minimum size'
\newcommand\compartText{\small}% to set the 'node' option 'font' (to one of: \tiny \scriptsize \footnotesize \small \normalsize \large \Large \LARGE \huge \Huge)
\newcommand\susComptDistFact{0.2}% to make distance between susceptible compartments a factor of \compartDist as follows:
\newcommand\susDist{\compartDist*\susComptDistFact}% to combine user-specified values to set the distance between S compartments
\newcommand\arrowFromBox{2pt}% to set distance between arrows and boxes
\newcommand\braceLineWidth{2pt}% to set the brace line width
% END colors, sizes, & distances

\begin{tikzpicture}
  [scale=1.3,
  node distance=0.5*\compartDist,
  bend angle=25,
  shorten <=\arrowFromBox,shorten >=\arrowFromBox,% so that arrows don't touch boxes
  % style common to all compartments
  C/.style={thick,inner sep=0pt,minimum size=\compartSize,font=\compartText},
  B/.style={C,shape=circle,draw=\col{\Bcol}{\drawIntens},fill=\col{\Bcol}{\fillIntens}},
  %S/.style={B,shape=rectangle,draw=\col{\Scol}{\drawIntens},fill=\col{\Scol}{\fillIntens},text=\col{\Scol}{\textIntens}},
  S/.style={B,shape=rectangle,draw=\col{\Scol}{\drawIntens},fill=\col{\Scol}{\fillIntens}},
  I/.style={S,draw=\col{\Icol}{\drawIntens},fill=\col{\Icol}{\fillIntens}},
  R/.style={S,draw=\col{\Rcol}{\drawIntens},fill=\col{\Rcol}{\fillIntens}},
  ghost/.style={B,draw=white,fill=white,opacity=0},
  params/.style={C,shape=ellipse,minimum height=\compartSize,minimum width=2*\compartSize,
  draw=\col{\Paramscol}{\drawIntens},fill=\col{\Paramscol}{\fillIntens}},
  fromto/.style={->,>=stealth',semithick,auto}]
   
   % begin figure
    \node[B] (B)                {$B_{\unit}$};
    \node[S] (S1) [right=of B]  {$S_{\unit,1}^B$}; \draw[fromto] (B) to (S1);
    \node[S] (S2) [right=\susDist of S1] {$S_{\unit,2}^B$}; \draw[fromto] (S1) to (S2);
    \node[S] (S3) [right=\susDist of S2] {$S_{\unit,3}^B$}; \draw[fromto] (S2) to (S3);
    \node[S] (S4) [right=\susDist of S3] {$S_{\unit,4}^B$}; \draw[fromto] (S3) to (S4);
    \node[ghost] (S) at ($(S3)!0.5!(S4)$) {};% ghost compartment for easier alignment of compartment IO
    \node[S] (S5) [right=\susDist of S4] {$S_{\unit,5}^B$}; \draw[fromto] (S4) to (S5);
    \node[S] (S6) [right=\susDist of S5] {$S_{\unit,6}^B$}; \draw[fromto] (S5) to (S6);
    \node[S] (SO) [right=of S6] {$S_{\unit}^{\hspace{0.2mm}O}$}; \draw[fromto] (S6) -- (SO);
    \node[I] (IB) [below=of S] {$I_{\unit}^B$};
    \draw[fromto,bend right] (S1) to (IB);%\draw[fromto,bend right]  (S1) to ([yshift=-2pt]IB.west);
    \draw[fromto,bend right] (S2) to (IB);
    \draw[fromto,bend right] (S3.south) to (IB.north west);
    \draw[fromto,bend left]  (S4.south) to (IB.north east);
    \draw[fromto,bend left]  (S5) to (IB);
    \draw[fromto,bend left]  (S6) to (IB);
    \node[I] (IO) [below=of SO] {$I_{\unit}^{\hspace{0.2mm}O}$}; \draw[fromto] (SO) to (IO);
    \node[ghost] (I) at ($(IB)!0.5!(IO)$) {};% ghost compartment for easier alignment of compartment R
    \node[I] (Y)  [right=of IO,shape=circle] {$Y_{\unit}$}; 
    %\draw[fromto,dashed] (IO) to node[auto] {$\rho$} (Y);% with label
    \draw[fromto,dashed] (IO) to (Y);
    \node[R] (R)  [below=of I]  {$R_{\unit}$}; \draw[fromto,bend right] (IB) to (R.west); \draw[fromto,bend left] (IO) to (R.east);
    
    \node[ghost] (param)  [below=of B] {};% ghost compartment for easier alignment of params
    \node[params] (params)  [below=of param] {$\phi, \psi_\unit$}; 
    \node (frame) [draw=black, fit= (B) (Y) (R) (params), inner sep=0.50cm, ultra thick] {};
        \node [yshift=3.0ex, xshift=-10.0ex] at (frame.south east) {\compartText $\bm{\unit \in 1:\Unit}\hspace{2mm}$};
    %\node (Y) [draw=blue, fit= (B) (Y) (R), inner sep=1.00cm, 
    %            dashed, ultra thick, fill=blue!10, fill opacity=0.2] {};
    %    \node [yshift=3.0ex, blue] at (Y.south) {Label B};
\end{tikzpicture}
\caption{Flow diagram for the polio panel model. 
Each individual resides in exactly one of the compartments denoted by square boxes.
Solid arrows represent possible transitions into a new compartment.
Circles represent observed variables: the reported incidence, $Y_{\unit}$, and births, $B_{\unit}$.
The dependence of $Y_\unit$ on $I^O_\unit$ is denoted by a dashed arrow.
Colors and rows represent disease status: unexposed is green (top row); infected is red (middle row); recovered is purple (bottom row).
The panel structure is indicated by the replication of this model over $\unit\in 1{\mycolon}\Unit$.
The shared parameter vector $\phi=(\rho,\sigma_{\mathrm{dem}},\psi,\tau)$ and unit-specific parameter $\psi_\unit=(b_{\unit,1:6},\sigma_{\mathrm{env},\unit},\tilde S^O_{\unit,0},\tilde I^O_{\unit,0})$ are identified in the gray ellipse.
}\label{fig:polio-diagr}
\end{figure}

\begin{figure}\label{fig:polio}
\begin{knitrout}
\definecolor{shadecolor}{rgb}{1, 1, 1}\color{fgcolor}

{\centering \includegraphics[width=4in]{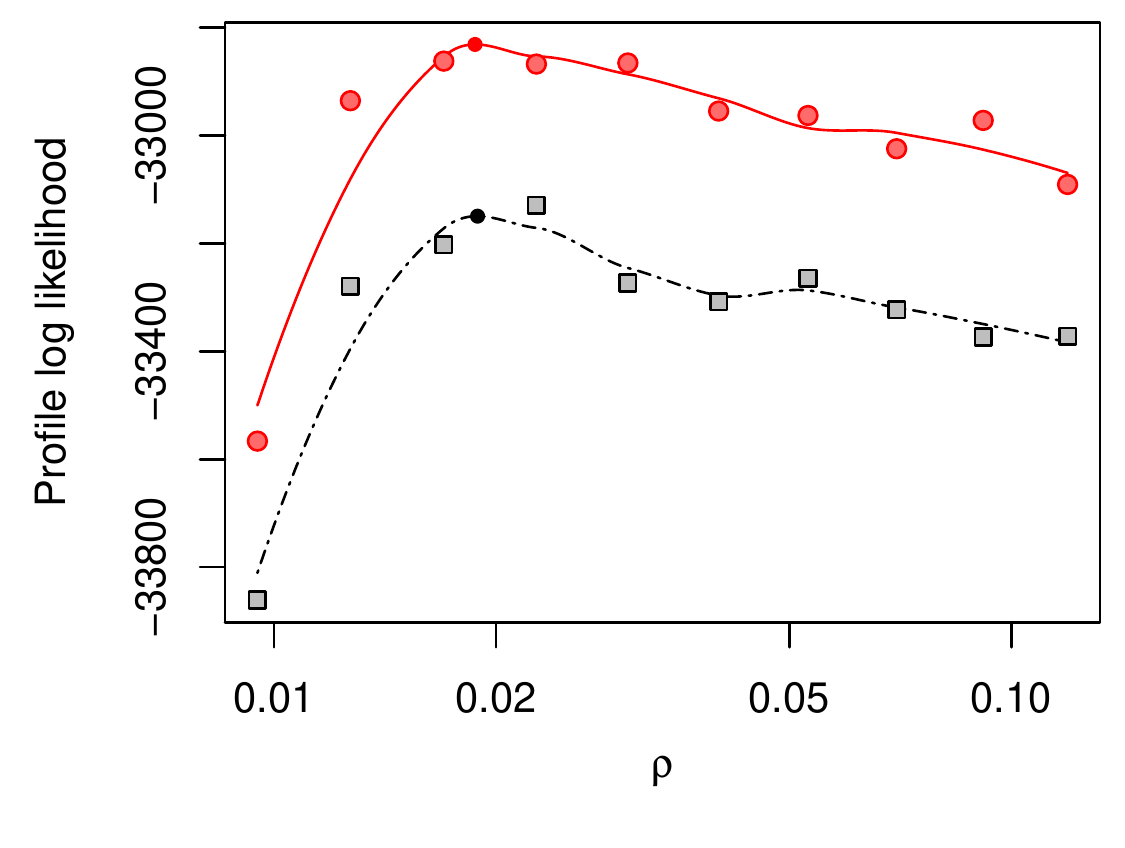} 

}

\end{knitrout}
\negBelowProfile
\caption{Profile log likelihood of $\rho$ for the polio model, computed with marginal maximization for unit-specific parameters (red circles and line) and without (black squares and line) with algorithmic parameters in table~\SuppTabAlgPars in the supplement.
Figure~\ref{fig:polio2} gives a closer look in a neighborhood of the maximum and constructs a confidence interval.}\label{fig:polio}
\end{figure}

The model of \citet{martinez-bakker15} places each individual in the population into one of ten compartments: susceptible babies in each of six one-month birth cohorts ($S^B_1$,...,$S^B_6$), susceptible older individuals ($S^O$), infected babies ($I^B$), infected older individuals ($I^O$), and individuals who have recovered with lifelong immunity ($R$).
Our {\PanelPOMP} version of their model has a latent process determining the number of individuals in each compartment at each time for each unit $\unit\in 1\mycolon\Unit$.
We write
\begin{equation} \nonumber
X_\unit(t)=\big(S^B_{\unit,1}(t),...,S^B_{\unit,6}(t), I^B_\unit(t),S^O_\unit(t),I^O_\unit(t),R_\unit(t) \big).
\end{equation}
The flows through the compartments are graphically represented in Figure~\ref{fig:polio-diagr}.
Births for each state $\unit$ are treated as a covariate time series, known from census data \citep{martinez-bakker14}.
Babies under six months are modeled as fully protected from paralytic polio, but capable of a gastro-intestinal polio infection.
Infection of an older individual leads to a reported paralytic polio case with probability $\rho_\unit$.

Since duration of infection is comparable to the one-month reporting aggregation, a discrete time model may be appropriate.
The model is therefore specified only at times $t_{\unit,n}=t_n=1932+ (4+n)/12$ for $n=0,\dots,249$.
We write
\begin{equation}\nonumber
X_{\unit,n}=X_\unit(t_n)=\big(S^B_{\unit,1,n},...,S^B_{\unit,6,n},S^O_{\unit,n}, I^B_{\unit,n},I^O_{\unit,n},R_{\unit,n} \big).
\end{equation}
The mean force of infection, in units of $\mathrm{yr}^{-1}$, is modeled as
\begin{equation}\nonumber
\bar\lambda_{\unit,n}=\left( \beta_{\unit,n} \frac{I^O_{\unit,n}+I^B_{\unit,n}}{P_{\unit,n}} + \psi \right),
\end{equation}
where $P_{\unit,n}$ is a census population covariate for state $u$ interpolated to time $t_n$ and seasonality of transmission is modeled as
\begin{equation}\nonumber
\beta_{\unit,n}=\exp\left\{ \sum_{k=1}^K b_{\unit,k}\, \xi_{k}(t_n) \right\},
\end{equation}
with $\{\xi_k(t),k=1,\dots,K\}$ being a periodic B-spline basis.
We set $K=6$. The force of infection has a stochastic perturbation,
$$\lambda_{\unit,n} = \bar\lambda_{\unit,n} \epsilon_{\unit,n},$$
where $\epsilon_{\unit,n}$ is a Gamma random variable with mean 1 and variance $\sigma^2_{\mathrm{env},\unit} + \sigma^2_{\mathrm{dem},\unit}\big[\lambda_{\unit,n}]^{-1}$.
These two terms capture variation on the environmental and demographic scales, respectively \citep{breto09}.
All compartments suffer a mortality rate, set at $\delta=[60\, \mathrm{yr}]^{-1}$ for all states.
Within each month, all susceptible individuals are modeled as having exposure to constant competing hazards of mortality and polio infection.
The chance of remaining in the susceptible population when exposed to these hazards for one month is therefore
$$p_{\unit,n} = \exp\left\{-\frac{\delta+\lambda_{\unit,n}}{12}\right\},$$
with the chance of polio infection being
$$q_{\unit,n} = (1-p_{\unit,n})\,\frac{\lambda_{\unit,n}}{\lambda_{\unit,n}+\delta}.$$
We employ a continuous population model: writing $B_{\unit,n}$ for births in month $n$ for state $\unit$, the full set of model equations is,
$$\begin{aligned}
 I^B_{\unit,n+1} &= q_{\unit,n} \sum_{k=1}^6 S^B_{\unit,k,n}, & \hspace{3mm} S^B_{\unit,1,n+1} &= B_{\unit,n+1},\hspace{8mm}  S^B_{\unit,k,n+1} &= p_{\unit,n}S^B_{\unit,k-1,n} \hspace{2mm} \mbox{for $k\in 1{\mycolon}6$,}\\
I^O_{\unit,n+1} &= q_{\unit,n} S^O_{\unit,n}, &
S^O_{\unit,n+1} &= p_{\unit,n}(S^O_{\unit,n}+S^B_{\unit,6,n}). 
\end{aligned}$$
The model for the reported observations, conditional on the latent process, is a discretized normal distribution truncated at zero, with both environmental and Poisson-scale contributions to the variance:
$$Y_{\unit,n}= \max\{\mathrm{round}(Z_{\unit,n}),0\}, \quad Z_{\unit,n}\sim\mathrm{Normal}\left(\rho I^O_{\unit,n}, \rho_\unit I^O_{\unit,n}+\left(\tau_\unit I^O_{\unit,n}\right)^2\right),$$
where $\mathrm{round}(x)$ obtains the integer closest to $x$.
Additional parameters are used to specify initial values for the latent process at time $t_0=1932+ 4/12$.
We will suppose there are parameters $\big(\tilde S^B_{\unit,1,0},...,\tilde S^B_{\unit,6,0}, \tilde I^B_{\unit,0},\tilde I^O_{\unit,0},\tilde S^O_{\unit,0}\big)$ that specify the population in each compartment at time $t_0$ via
$$ S^B_{\unit,1,0}= {\tilde S}^B_{\unit,1,0} ,...,S^B_{\unit,6,0}= \tilde S^B_{\unit,6,0}, \quad I^B_{\unit,0}= P_{\unit,0} \tilde I^B_{\unit,0},\quad S^O_{\unit,0}= P_{\unit,0} \tilde S^O_{\unit,0}, \quad I^O_{\unit,0}= P_{\unit,0} \tilde I^O_{\unit,0}.$$
The initial conditions are simplified by ignoring infant infections at time $t_0$.
Thus, we set $\tilde I^B_{\unit,0}=0$ and use monthly births in the preceding months (ignoring infant mortality) to fix $\tilde S^B_{\unit,k,0}=B_{\unit,1-k}$ for $k=1,\dots,6$.
The estimated initial conditions for state $\unit$ are then defined by the two parameters $\tilde I^O_{\unit,0}$ and $\tilde S^O_{\unit,0}$, since the initial recovered population, $R_{\unit,0}$, is specified by subtraction of all the other compartments from the total initial population, $P_{\unit,0}$.

%*% run polio2a example 

%*% run polio2b example 

%*% run polio2c example 

\begin{figure}\label{fig:polio2}
\begin{knitrout}
\definecolor{shadecolor}{rgb}{1, 1, 1}\color{fgcolor}

{\centering \includegraphics[width=3.3in]{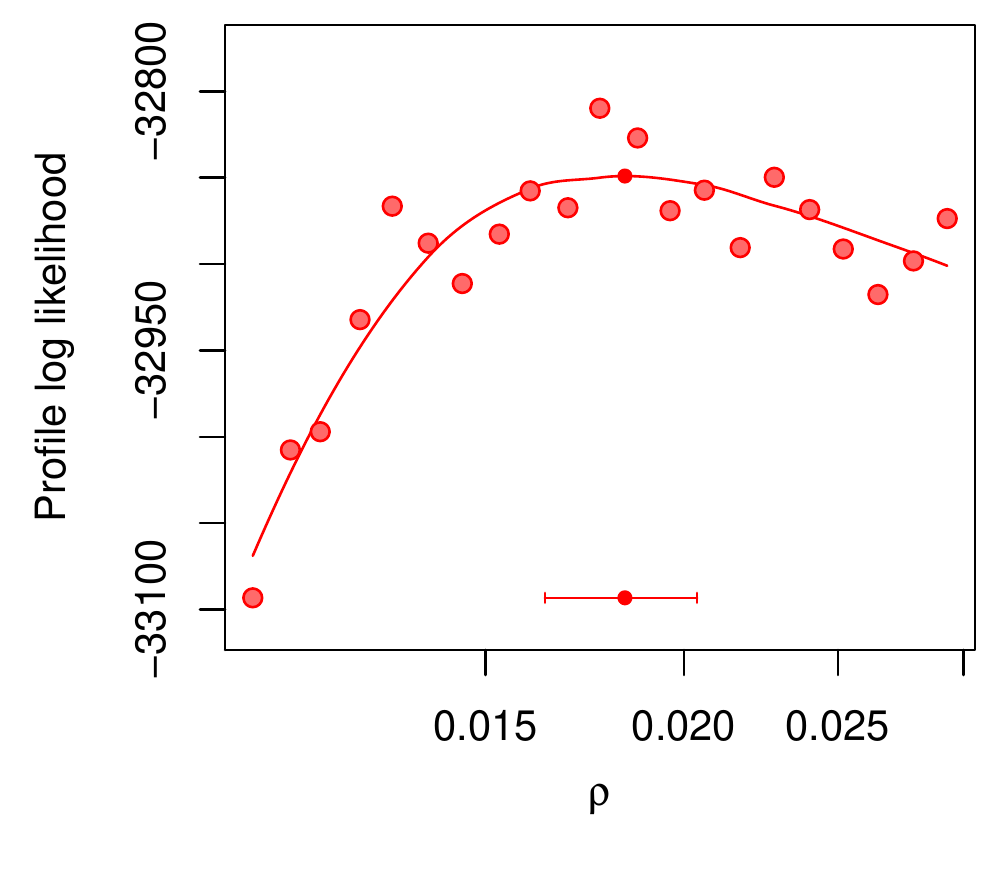} 

}

\end{knitrout}
\negBelowProfile
\caption{Profile log likelihood of $\rho$ for the polio model, computed in a neighborhood of the MLE.
The horizontal bar shows a $95$\% MCAP confidence interval with a small filled circle marking the MLE obtained with algorithmic parameters in table~\SuppTabAlgPars in the supplement.
}\label{fig:polio2}
\end{figure}

Figure~\ref{fig:polio} shows the profile likelihood for the shared reporting rate, $\rho$, evaluated across a wide interval to investigate large-scale features of the likelihood surface.
Including a marginal maximization step in PIF leads to gains in agreement with the findings of Figure~\ref{fig:gompertz}.
Figure~\ref{fig:polio} indicates an MLE around 0.02 and so this parameter range was studied further in a higher resolution profile in Figure~\ref{fig:polio2}.
On this localized plot, we can see Monte Carlo error of order 10 log units in the maximization and evaluation of the log likelihood.
Since construction of this plot employed
$528.0$
core days of computational effort, we had limited capacity for further reductions in Monte Carlo error by further increases in computation.
Fortunately, MCAP methodology is able to handle Monte Carlo error on this scale: see, for example, the noisiest profile in Figure~\ref{fig:gompertz} using only joint maximization.
The resulting $95$\% confidence interval from Figure~\ref{fig:polio2} is $(0.016,0.020)$, which is consistent with estimates for the fraction of polio infections leading to acute paralysis in USA in this era \citep{melnick53}.
By contrast, \citet{martinez-bakker15} found point estimates ranging from 0.0025 to 0.03 when analyzing each state separately, with wide confidence intervals evident from the profiles in their Figures~S9--S17.

%% begin{shared-comm.tex}
\newcommand\othersFail{Likelihood-based inference for data on this scale ($U=49$, $N_u=249$) has been considered intractable for general nonlinear PanelPOMP models using previous methodology. Indeed, even for a single observed time series, inference for general nonlinear POMP models has only recently become routine \citep{king16}. Evidently, the difficulties are a result of model complexity as much as the sheer volume of data. The interaction of model complexity with a modest increase in data size is the current challenge.}%
%% end{shared-comm.tex}
\revision{\othersFail}

\section{Dynamic variation in sexual contact rates}\label{sec:hiv}

%*% run contacts example 

We demonstrate PIF for analysis of panel data on sexual contacts, using the model and data of \citet{romero-severson15}.
The data consist of many short time series, a common situation in classical longitudinal analysis.
We show that PIF provides useful flexibility to permit consideration of scientifically relevant dynamic models including latent dynamic variables. 

Basic population-level disease transmission models suppose equal contact rates for all individuals in a population \citep{keeling08}.
Sometimes these models are extended to permit heterogeneity between individuals.
Heterogeneity within individuals over time has rarely been considered, yet, there have been some indications that dynamic behavioral change plays a substantial role in the HIV epidemic.
\citet{romero-severson15} quantified dynamic changes in sexual behavior by fitting a model for dynamic variation in sexual contact rates to panel data from a large cohort study of HIV-negative gay men \citep{vittinghoff99}.
Here, we analyze the data on total sexual contacts over $N_\unit=4$ consecutive 6-month periods for the 
$\Unit=882$ men in the study who had no missing observations.
A sample of the time series in this panel is plotted in the supplement.

For behavioral studies, we interpret ``mechanistic model'' broadly to mean a mathematical model describing phenomena of interest via interpretable parameters.
In this context, we want a model that can describe (i) differences between individuals; (ii) differences within individuals over time; (iii) flexible relationships between mean and variance of contact rates.
\citet{romero-severson15} developed a {\PanelPOMP} model capturing these phenomena.
Suppose that each individual $\unit\in 1\mycolon\Unit$ has a latent time-varying rate $\contactsRate_\unit(t)$ of making a sexual contact.
Each data point, $y^*_{\unit,n}$, is the number of reported contacts for individual $\unit$ between time $t_{\unit,n-1}$ and $t_{\unit,n}$.
Integrating the unobserved process $\{\contactsRate_\unit(t)\}$ gives the conditional expected value in~\eqref{eq:contacts-dmeas} of contacts for individual $\unit$ in reporting interval $n$, via
$$C_{\unit,n}= \alpha^{n-1}\int_{t_{\unit,n-1}}^{t_{\unit,n}} \contactsRate_\unit(t)\, dt,$$
where $\alpha$ is an additional secular trend that accounts for the observed decline in reported contacts.
A basic stochastic model for homogeneous count data would model $y^*_{\unit,n}$ as a Poisson random variable with mean and variance equal to $C_{\unit,n}$
\citep{keeling08}.
However, the variance in the data is much higher than the mean \citep{romero-severson12}.
Negative binomial processes provide a route to modeling dynamic over-dispersion \citep{breto11}.
Here, we suppose that
\begin{equation}\label{eq:contacts-dmeas}
Y_{\unit,n} | C_{\unit,n},D_\unit \sim \mathrm{NegBin}\, \left(C_{\unit,n},D_{\unit}\right),
\end{equation}
a conditional negative binomial distribution with mean $C_{\unit,n}$ and  variance $C_{\unit,n}+C_{\unit,n}^2[D_\unit]^{-1}$.
Here, $D_\unit$ is called the dispersion parameter, with the Poisson model being recovered in the limit as $D_\unit$ becomes large.
The dispersion, $D_\unit$, can model increased variance compared to the Poisson distribution for individual contacts, but does not result in autocorrelation between measurements on an individual over time, which is observed in the data.
To model this autocorrelation, we suppose that individual $\unit$ has behavioral episodes of exponentially distributed duration within which $\{\contactsRate_\unit(t)\}$ is constant, but the individual enters new behavioral episodes at rate $\mu_R$.
At the start of each episode, $\{\contactsRate_\unit(t)\}$ takes a new value drawn from a Gamma distribution with mean $\mu_X$ and standard deviation $\sigma_X$.
Therefore, at each time $t$,
$$\contactsRate_\unit(t)\sim \mbox{Gamma}(\mu_X, \sigma_X).$$
To complete the model, we also assume $D_\unit \sim \mbox{Gamma}(\mu_D, \sigma_D)$.
The parameter vector is $\theta=(\mu_X, \sigma_X, \mu_D,\sigma_D, \mu_R,\alpha)$.

\begin{figure}\label{fig:contacts}
%\begin{center}
\begin{knitrout}
\definecolor{shadecolor}{rgb}{1, 1, 1}\color{fgcolor}

{\centering \includegraphics[width=3in]{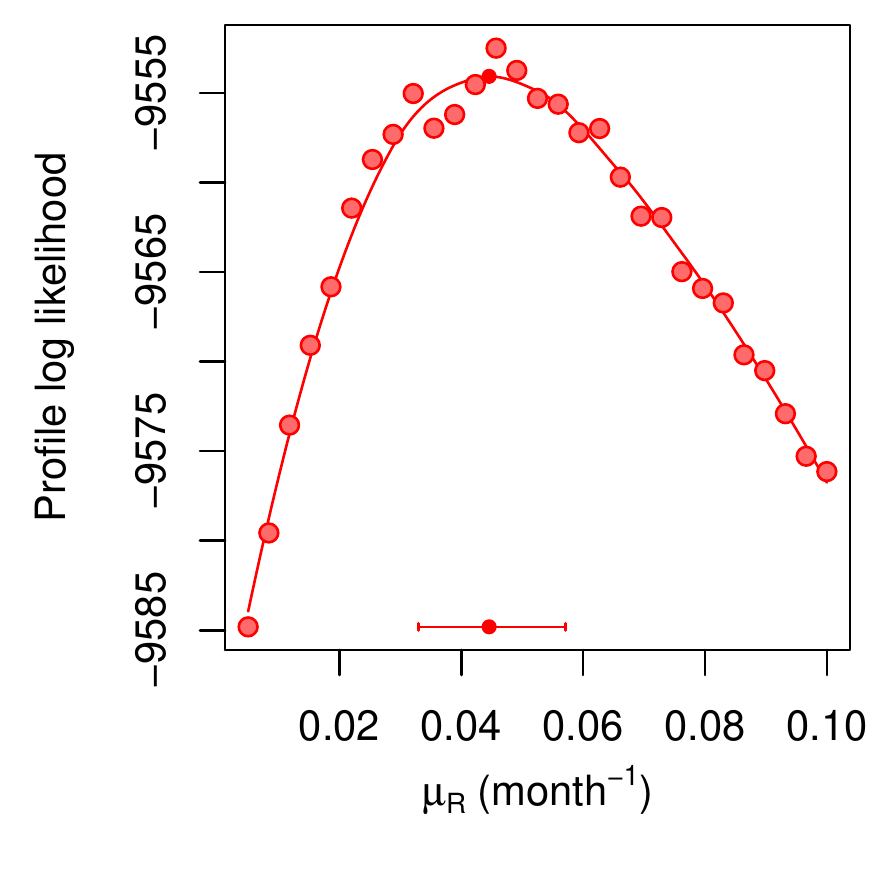} 

}

\end{knitrout}
%\end{center}
\negBelowProfile
\caption{Profile log likelihood of $\mu_R$ for a panel of size $\Unit=882$ for the contacts model.
The horizontal bar shows a $95$\% MCAP confidence interval with a small filled circle marking the MLE obtained with algorithmic parameters in table~\SuppTabAlgPars in the supplement.}\label{fig:contacts}
\end{figure}

Figure~\ref{fig:contacts} constructs a profile likelihood confidence interval for $\mu_R$. 
This result is comparable to Web Figure~1 of \citet{romero-severson15}, however, here we have shown how this example fits into a general methodological framework.
The profile demonstrates an intermediate level of computational challenge between the relatively simple Gompertz example of Section~\ref{sec:gompertz} and the extensive data, complex model, and correspondingly larger Monte Carlo computations of Section~\ref{sec:polio}.
Figure~\ref{fig:contacts} took
$68.7$
core days to compute.
No marginal maximization was necessary for this example, since all parameters were shared between all units.

\section{Discussion} \label{sec:dis}

When panel data are short, relative to the complexity of the model under consideration, there may be little information in the data about each unit-specific parameter.
In such cases, it can be appropriate to replace some unit-specific parameters by latent random variables.
By analogy with linear regression analysis, these unit-specific latent random variables are called {\it random effects}, and unit-specific parameters treated as unknown constants are called {\it fixed effects}.
Models with random effects are also called {\it hierarchical} models, since an additional hierarchy of modeling is required to describe the additional latent variables, and parameters of the random effect distribution are consequently termed {\it hyperparameters}.
From the perspective of statistical inference, using random effects reduces the number of fitted parameters, at the expense of adding additional modeling assumptions.
From the perspective of computation, random effects reduce the dimensionality of the likelihood optimization challenge, while increasing the dimensionality of the latent space which must be integrated over to evaluate the likelihood.
The use of random effects provides an opportunity for the estimation of individual unit-specific effects to borrow strength from other panel units, via estimation of the hyperparameters.
Therefore, random effects can have particular value if one is interested in estimating the unit-specific effects.
However, when the research question is focused on shared parameters or higher-level model structure decisions such as whether a parameter should be included in the model, the individual unit-specific parameters can be a distraction.
Rather than spending time developing and justifying a distribution for the random effects, simpler statistical reasoning can be obtained by avoiding these issues and employing fixed effects.

The sexual contact model has random effects $D_\unit$ and has no fixed effects.
As discussed above, this is appropriate for panel data with very short time series.
By contrast, the polio data are relatively long time series, enabling the use of fixed effects.

Panel time series analysis shares similarities with functional data analysis \citep{ramsay97}. 
Within functional data analysis, a representation of dynamic mechanisms can be incorporated via principal differential analysis \citep{ramsay96}.
Partially observed stochastic dynamic models are not within the usual scope of the field of functional data analysis, though there is no need for a hard line separating functional data analysis from panel data analysis.
   
We wrote an R package {\texttt{panelPomp}} (available at \url{https://github.com/cbreto/panelPomp}) that provides a software environment for developing methodology and data analysis tools for {\PanelPOMP} models extending the {\texttt{pomp}} package \citep{king16}.
The implementation of PIF in {\texttt{panelPomp}} was used for the results of this paper.
PIF and the {\texttt{panelPomp}} package have already proved useful for scientific investigations \citep{ranjeva17,ranjeva18}.

Iterated filtering algorithms provide an approach to plug-and-play full-information likelihood-based inference that has been applied in challenging nonlinear mechanistic time series analyses, especially in epidemiology \citep[reviewed by][]{breto18}.
Reduced information methods, such as those using simulations to compare the data with a collection of summary statistics, can lead to substantial losses in statistical efficiency \citep{fasiolo16,shrestha11}.
Particle Markov chain Monte Carlo \citep{andrieu10} provides a route to plug-and-play full-information Bayesian inference, but the methodology requires computational feasibility of log likelihood estimates with a standard error of around 1 log unit \citep{doucet15pmcmc}.
PIF is the first plug-and-play full-information likelihood-based approach demonstrated to be applicable on general partially observed nonlinear stochastic dynamic models for panel data analysis on the scale we have considered.

\newpage
\setcounter{equation}{0}
\setcounter{figure}{0}
\setcounter{page}{1}
\setcounter{section}{0}
\setcounter{table}{0}

%% begin{add from supp.tex}
\noindent {\Large \bf Supplementary Content}
\renewcommand{\refname}{Supplementary References}
\renewcommand\thefigure{S-\arabic{figure}}
\renewcommand\thetable{S-\arabic{table}}
\renewcommand\thepage{S-\arabic{page}}
\renewcommand\thesection{S\arabic{section}}
\renewcommand\theequation{S\arabic{equation}}

\newcommand\constantell{\ell_{\begin{picture}(-2,2)(-2,-2)\circle*{2}\end{picture}}}

\section{Alternative POMP representations of a PanelPOMP}

Section~2.1 of the main text developed a partially observed Markov process (POMP) representation of a {\PanelPOMP}, which we call construction~R1.
The following constructions, {R\ref{R1}} and {R\ref{R2}}, provide two alternative ways to write a PanelPOMP as a POMP. 

\newcounter{R}
\renewcommand\listR{R}

\begin{list}{(\listR\arabic{R})}{\usecounter{R}\conditionList\stepcounter{R}}

\item\label{R1}
For a panel in which each unit is observed over the same time interval, we can write
$$ X^{[\ref{R1}]}(t) = \big(X_{1}(t),X_2(t),\dots,X_\Unit(t)\big).$$
This constructs a POMP by concatenating the latent state vectors for each separate unit of the PanelPOMP.
The dimension of the resulting latent process increases with the number of panel units, $\Unit$.
Sequential Monte Carlo (SMC) methods struggle with high-dimensional latent processes \citep{bengtsson08}.
This representation is therefore anticipated to be useful for SMC based methodology only when $\Unit$ is small.

\vspace{3mm}

\item\label{R2}
The latent process for a {\PanelPOMP} model need only be specified at the observation times. Therefore, we can define an equivalent integer-time POMP model,
$$ X^{[\ref{R2}]}(\unit)=\big(X_{\unit,0},X_{\unit,1},\dots,X_{\unit,N_\unit}\big).$$
The dynamics in this POMP model are trivial: $X^{[\ref{R2}]}(i)$ is independent of $X^{[\ref{R2}]}(j)$ for $i\neq j \in 1\mycolon \Unit$.
Due to the `curse of dimensionality' for importance sampling, this representation is useful for SMC based methodology only when all of $N_1,\dots,N_\Unit$ are small.
This representation can provide a simple way to apply existing POMP methodology to panel data, and for that reason it was adopted by \citet{romero-severson15}.
\end{list}

The only reasons of which we are aware to give preference to {R\ref{R1}} or {R\ref{R2}} over R1 are  small potential gains in conceptual and coding simplicity.
However, the scaling difficulties faced by both {R\ref{R1}} and {R\ref{R2}} make them inappropriate for general-purpose methodology and software based on SMC.

\section{Estimators for the likelihood of a \PanelPOMP}

Consider $\Nrep \geq 2$ independent particle filters, each with $\Np \geq 1$ particles, which give independent Monte Carlo likelihood estimators $\mclik{\unit}{\nrep}, \nrep\in 1{\mycolon}\Nrep$, for each unit.
We work with a constant parameter value $\theta$ and write $\ell_\unit(\theta)=\ell_\unit$. 
The Monte Carlo estimator is unbiased and has finite variance, written as
\begin{equation}\nonumber
\E[\mclik{\unit}{\nrep}] = \ell_{\unit},  \qquad  \var[\mclik{\unit}{\nrep}] = \sigma^2_{\unit} < \infty.
\end{equation}
A corresponding estimator of the full panel likelihood based on replication $\nrep$ is 
\begin{equation}\nonumber
\mclik{}{\nrep} = \prod_{\unit=1}^{\Unit} \mclik{\unit}{\nrep},
\end{equation}
which has mean and variance given by
\begin{equation}\nonumber
\E[\mclik{}{\nrep}] = \ell \, , 
\hspace{15mm}
\var[\mclik{}{\nrep}] 
= \prod_{\unit=1}^{\Unit} \Bigl\{ \sigma^2_{\unit} + \ell_{\unit}^2 \Bigr\} - \ell^2.
\end{equation}
A natural approach to combining the $\Nrep$ independent likelihood estimators for estimation of the likelihood of a single panel unit $\unit$ is
\begin{equation}\nonumber
\mcUnit_{\unit} = \frac{1}{\Nrep}\sum_{\nrep=1}^{\Nrep}{\mclik{\unit}{\nrep}},
\end{equation}
which has mean and variance given by
\begin{equation}\nonumber
 \E[\mcUnit_{\unit}] = \ell_{\unit} \, ,
\hspace{15mm} 
\var[\bar{L}_{\unit}] = \Nrep^{-1}\sigma^2_{\unit}. 
\end{equation}
However, it is not immediately clear how to combine the unit-level likelihood estimators to estimate the likelihood of the entire panel.
We consider two estimators,
\begin{equation}\nonumber
\likA = \frac{1}{\Nrep}\sum_{\nrep=1}^{\Nrep}{\mclik{}{\nrep}},
\hspace{20mm}
\likB = \prod_{\unit=1}^\Unit \mcUnit_{\unit}.
\end{equation}
While both estimators are unbiased, $\likA$ is less efficient than $\likB$.
To see this, consider first their variances,
\begin{eqnarray}
\label{eq:varA}\var[\likA] &=& \frac{1}{\Nrep}\left[
\,
\prod_{\unit=1}^{\Unit} \Bigl\{ \sigma^2_{\unit} + \ell_{\unit}^2 \Bigr\} - \ell^2
\right]
\\
\nonumber\var[\likB] &=& \E\left[ 
\,
\prod_{\unit=1}^\Unit \mcUnit_{\unit}
  \right]^2 - \ell^2
\\
\label{eq:varB}&=& \prod_{\unit=1}^\Unit \Bigl\{
  \frac{\sigma^2_\unit}{\Nrep} + \ell_\unit^2
  \Bigr\}  - \ell^2.
\end{eqnarray}
Expanding the product in \eqref{eq:varA} yields
\begin{eqnarray} \label{eq:2varA}
\var[\likA] &=& \frac{1}{\Nrep}\left[
\sum_{k_{1:\Unit}\in\{0,1\}^\Unit}
\left\{
  \prod_{\unit=1}^\Unit
    \sigma_\unit^{2k_\unit} \ell_\unit^{2(1-k_\unit)}
\right\} - \ell^2
\right].
\end{eqnarray}
The term $\ell^2$ in \eqref{eq:2varA} cancels with the summand for $k_{1:\Unit}=(0,0,\dots,0)$, giving
\begin{eqnarray}\label{likA:expansion}
\var[\likA] &=&
\sum_{k_{1:\Unit}\in\{0,1\}^\Unit\backslash \{0\}^\Unit}
\left\{
\,
\frac{1}{\Nrep}
\,
  \prod_{\unit=1}^\Unit
    \sigma_\unit^{2k_\unit} \ell_\unit^{2(1-k_\unit)}
\right\}.
\end{eqnarray}
An analogous expression for $\var[\likB]$, derived from \eqref{eq:varB}, is
\begin{eqnarray}\label{likB:expansion}
\var[\likB] &=&
\sum_{k_{1:\Unit}\in\{0,1\}^\Unit\backslash \{0\}^\Unit}
\left\{
  \left[\,
  \prod_{\unit=1}^\Unit
  \frac{1}{\Nrep^{k_\unit}}
\right]
  \prod_{\unit=1}^\Unit
    \sigma_\unit^{2k_\unit} \ell_\unit^{2(1-k_\unit)}
\right\}.
\end{eqnarray}
Comparing equivalent terms in the sums for~\eqref{likA:expansion} and~\eqref{likB:expansion} we see that, supposing that either variance is strictly positive (which, from the unbiasedness of the likelihood estimator, implies that $\sigma^2_\unit$ and $\ell_\unit(\theta)$ are strictly positive) and given that $\Nrep>1$,
\begin{equation}\nonumber
\var[\likB] <\var[\likA].
\end{equation}
For a quantitative comparison of $\likA$ and $\likB$, consider the situation with constant likelihood
\[\ell_\unit=\constantell\]
and constant variance, $\sigma^2_\unit=\sigma^2$. 
Then,
\begin{equation} \label{lik:comparison}
\var[\likA]= \frac{1}{\Nrep}\Big[
 \big(\sigma^2+\constantell^2\big)^\Unit - \constantell^{2\Unit}
\Big],
\hspace{15mm}
\var[\likB] = \Big(
\frac{\sigma^2}{\Nrep} + \constantell^2
\Big)^\Unit - \constantell^{2\Unit}.
\end{equation}
Further, suppose we are interested in the relative likelihood so we can scale to $\constantell=1$. 
Now, if $\Nrep=c\Unit$ for some constant $c$, then we see from \eqref{lik:comparison} that $\var[\likB]$ is stable as $\Unit\to\infty$ whereas $\var[\likA]$ increases exponentially. 

\section{Considerations for likelihood shortfall}

%% begin{shared-comm.tex}
\newcommand\shortfallA{The shortfall is expected to be hard to calculate outside of toy problems. To quantify this, we would have to know the actual profile likelihood function, but the motivation for using this methodology is that it provides us with the best available approximation to this function. Nevertheless, theoretical and empirical approaches can give us some relevant insights.}%
\newcommand\shortfallB{As a profile interval becomes narrower, the variation of the bias across the confidence interval converges to zero since the Monte Carlo optimization and likelihood evaluation problems solved for each profile point become increasingly similar. In the limit when the confidence interval collapses toward a single point, the bias becomes trivially constant over the relevant part of the profile. The interval width for a shared parameter in large panels can be expected to be short, since the accumulated information over a large number of panels is anticipated to result in profiles with tall narrow peak. On the other hand, the level of Monte Carlo noise in small panels can be expected to be closer to that in multivariate time series. For time series analysis, a body of existing empirical analyses \citep[reviewed by][]{breto18} suggests it is often feasible to apply sufficient computational effort to make Monte Carlo error small, avoiding the need for any Monte Carlo adjustment.}%
\newcommand\shortfallC{From an empirical perspective, simulation studies can address coverage of the constructed confidence intervals. Simulations can provide a combined assessment of all the assumptions and computational approximations involved in a methodology. We have not focused on this holistic assessment in the current article, since we are primarily concerned with how to compute likelihood-based inferences rather than the topic of statistical properties such as coverage, size and power for likelihood-based inferences. This has led us to present results that check the particular computational issues of profile likelihood shortfall, in a specific case where this can readily be done, rather than considering overall statistical performance.}%
%% end{shared-comm.tex}
\revision{\shortfallA}

\revision{\shortfallB}

\revision{\shortfallC}

\section{Model misspecification and model selection}

%% begin{shared-comm.tex}
\newcommand\modelDiagnosticsAA{The generality of the PIF algorithm gives the scientist many options for diagnosing model misspecification and developing improved models. Since PIF is not constrained to a particular model, the data analyst is encouraged to consider and compare a wide range of models. Likelihood-based techniques such as Akaike's information criterion and likelihood ratio tests are available for model selection. Models can also be compared by assessing whether simulations from the fitted model capture features of interest.}%
\newcommand\modelDiagnosticsA{As PIF builds on sequential Monte Carlo (SMC), existing diagnostic methods for this methodology are available. A widely used SMC diagnostic tool is to compute an effective Monte Carlo sample size for each data point \citep{king16}. Observations with low effective sample size may be outliers, or may be hard to predict for some other reason. Effective sample size also plays a role for diagnosing successful Monte Carlo convergence, since one of the symptoms of an outlier is that the measurement is rare under the postulated model and so a large Monte Carlo sample is needed to accommodate the unexpected observation.}%
\newcommand\modelDiagnosticsB{Finding models with appropriate stochasticity to explain the data can be a critical aspect of effective data analysis. Both the latent process and the measurement model are open to critical assessment and improvement within the general PanelPOMP framework permitted by PIF. Some relevant issues on this are discussed by \citet{he10} and \citet{breto11} in the context of time series analysis.}%
\newcommand\modelDiagnosticsC{For panel data, the growing size of datasets can add difficulty to viewing diagnostic plots; sometimes one must look for creative summary statistics of the full set of diagnostics for each data point. The \texttt{panelPomp} R package provides a software environment to facilitate model development and diagnostics, extending the capabilities of \texttt{pomp} \citep{king16}.}%
%% end{shared-comm.tex}
\revision{\modelDiagnosticsAA}

\revision{\modelDiagnosticsA}

\revision{\modelDiagnosticsB}

\revision{\modelDiagnosticsC}

\section{Graphs for subsets of the polio and contacts data}

\begin{figure}[H]
\begin{knitrout}
\definecolor{shadecolor}{rgb}{1, 1, 1}\color{fgcolor}

{\centering \includegraphics[width=5in]{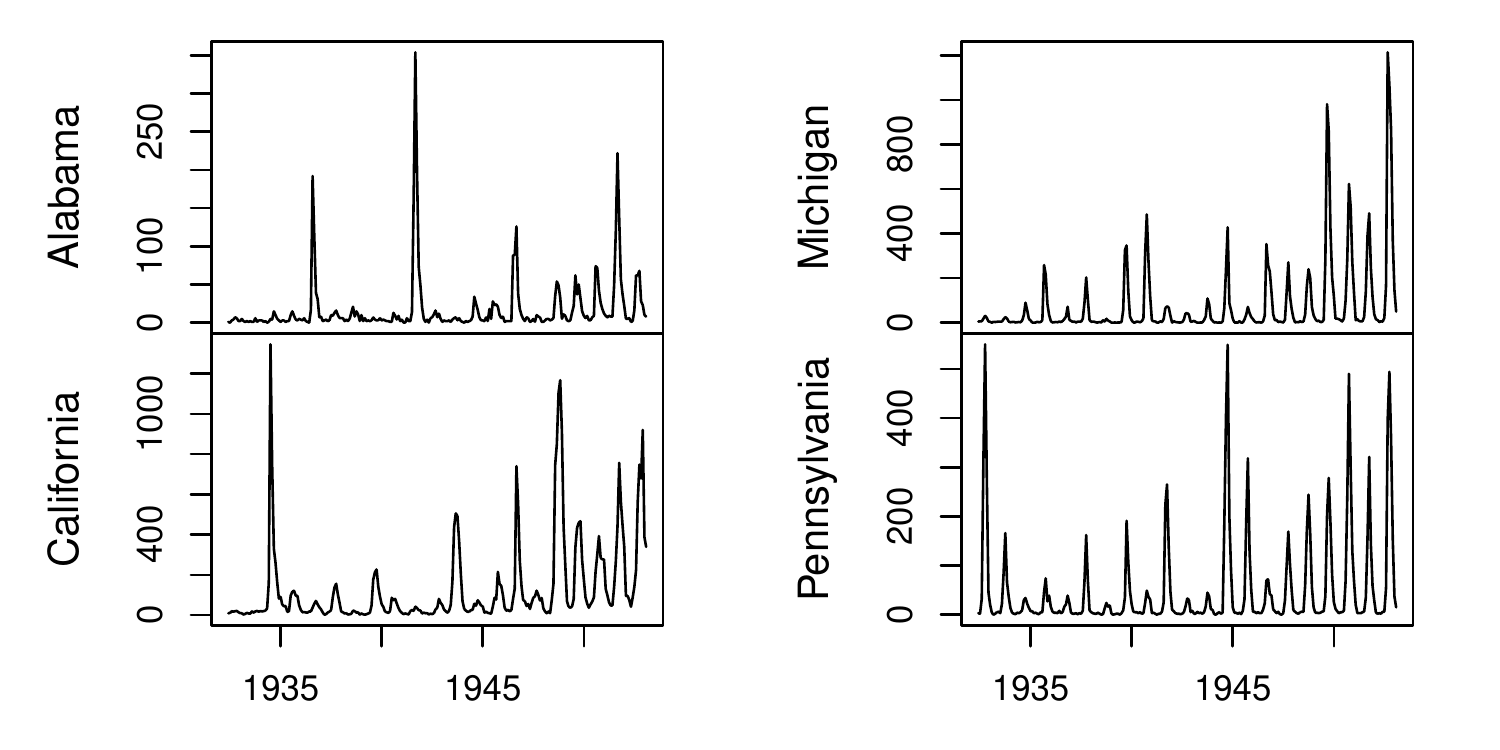} 

}

\end{knitrout}
%\vspace{-2mm}
\caption{Selection of $4$ time series from the panel dataset from \cite{martinez-bakker15} giving USA monthly of acute paralysis from polio from May 1932 through January 1953 for the 48 continuous US states and Washington D.C.
Birth data are missing for South Dakota (before January 1933) and Texas (before January 1934) and so these states were modeled over a reduced time interval.
The full data can be accessed from the \texttt{panpol} panelPomp object included in the \texttt{panelPomp} package.}\label{fig:polio-data}
\end{figure}

\begin{figure}[H]
\begin{knitrout}
\definecolor{shadecolor}{rgb}{1, 1, 1}\color{fgcolor}

{\centering \includegraphics[width=3.5in]{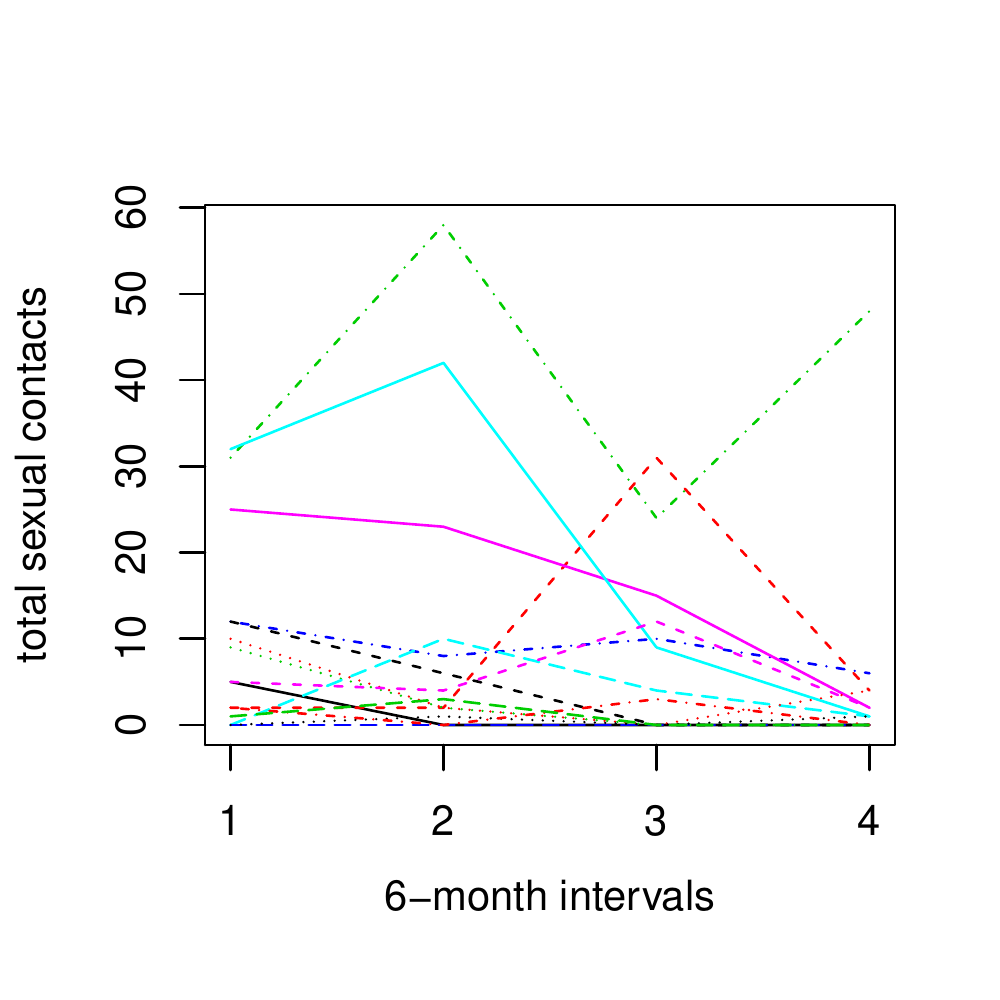} 

}

\end{knitrout}
\vspace{-8mm}
\caption{Sample of $15$ time series in the panel dataset from \cite{vittinghoff99} on total sexual contacts over four consecutive 6-month periods, for the $882$ men having no missing observations. The full data can be accessed from the \texttt{pancon} panelPomp object included in the \texttt{panelPomp} package.}\label{fig:contacts-data}
\end{figure}

%%AAAAAAAAAAAA

\section{Algorithmic parameters}

The following tables specify all the algorithmic parameters used for our examples.
%% begin{shared-comm.tex}
\newcommand\algparsDiagnosis{Algorithmic parameters were chosen by assessing diagnostic plots, together with consideration of total run time and quantification of Monte Carlo error in the final results. Diagnosis via convergence plots and effective sample size plots was carried out as described by \citet{king16}.}%
%% end{shared-comm.tex}
\revision{\algparsDiagnosis}

\begin{table}[H]
\begin{center}
\begin{tabular}{ccccc}
\toprule
 & Gompertz & polio & polio (MCAP interval) & contacts \\ 
  \midrule
\Np$_{\mathrm{pf}}$ & 4000 & 5000 & 5000 & 4000 \\ 
  \Nrep$_{\mathrm{pf}}$ & 10 & 10 & 10 & 10 \\ 
  \Np$_{\mathrm{if}}$ & 2000 & 4000 & 4000 & 4000 \\ 
  \Nrep$_{\mathrm{if}}$ & 13 & 19 & 27 & 13 \\ 
  \Nmif & 100 & 236 & 236 & 200 \\ 
  \Np$_{\mathrm{if},\mathrm{\unit}}$ & 1000 & 6000 & 6000 & -- \\ 
  \Nrep$_{\mathrm{pf},\mathrm{\unit}}$ & 4 & 2 & 3 & -- \\ 
  \Nmif$_{\mathrm{\unit}}$ & 50 & 118 & 118 & -- \\ 
  $\lambdaMCAP$ & 0.9 & 0.6 & 0.9 & 0.9 \\ 
   \bottomrule

\end{tabular}
\end{center}
\caption{Algorithmic parameter values used to produce plots in the main text.
\Np$_{\mathrm{pf}}$ particles were used for each of \Nrep$_{\mathrm{pf}}$ replicated particle filter Monte Carlo log likelihood estimates. 
These replicates were averaged using $\likB$ (defined in Section~\SuppSecLikAvg).
The resulting log likelihood estimates correspond to parameter values with maximum likelihood that were reached initializing the joint step of the panel iterated filtering algorithm at \Nrep$_{\mathrm{if}}$ different parameter starting values ($\Theta^0_{1:J}$ in the description of the PIF algorithm in the main text), running the algorithm for {\Nmif} iterations.
These \Nrep$_{\mathrm{if}}$ convergence points from the joint step were used to initialize \Nrep$_{\mathrm{pf},\mathrm{\unit}}$ marginal steps with \Nmif$_{\mathrm{\unit}}$ iterations and \Np$_{\mathrm{if},\mathrm{\unit}}$ particles. \revision{Monte Carlo profiles were obtained by applying loess smoothing to the profile evaluations with smoothing parameter $\lambdaMCAP$. A smaller value of $\lambdaMCAP$ was used for the initial exploratory polio profile than for the MCAP confidence intervals.}
}\label{tab:algpars}
\end{table}

\begin{table}[H]
\begin{center}
\begin{tabular}{ccccc}
\toprule
 & lower bound & upper bound & $\sigma_0$ & $\sigma_{\mathrm{\unit},0}$ \\ 
  \midrule
r & 0.05 & 0.20 & 0.00125 & -- \\ 
  $\tau_\unit$ & 0.05 & 0.20 & 0.05000 & 0.05 \\ 
   \bottomrule

\end{tabular}
\end{center}
\caption{Starting values, parameter transformations and perturbation specifications for applying PIF to the Gompertz model.
The first two columns give the lower and upper bounds of a hyper-rectangle sampled uniformly to generate a value used to initialize all particles $\Theta^0_{1:J}$ for each independent PIF replicate.
For the joint maximization, the perturbation sequence used was $\sigma_m=\sigma_0 0.5^{m/50}$.
For the marginal maximization, $\sigma_m=\sigma_{\unit,0}0.25^{m/50}$ was used instead.
These random perturbations were carried out as Gaussian random walks after applying a logarithmic transformation to ensure non-negativity constraints were met.
}\label{tab:gompertz}
\end{table}

\begin{table}[H]
\begin{center}
\begin{tabular}{cccccc}
\toprule
 & lower bound & upper bound & transformation & $\sigma_0$ & $\sigma_{\mathrm{\unit},0}$ \\ 
  \midrule
$\sigma_{\mathrm{dem}}$ & 0.0 & 0.50 & log & 0.02 & -- \\ 
  $\psi$ & 0.0 & 0.10 & log & 0.02 & -- \\ 
  $\tau$ & 0.0 & 0.10 & log & 0.02 & -- \\ 
  $b_{\unit,1}$ & -2 & 8.00 & -- & 0.02 & 0.02 \\ 
  $b_{\unit,2}$ & -2 & 8.00 & -- & 0.02 & 0.02 \\ 
  $b_{\unit,3}$ & -2 & 8.00 & -- & 0.02 & 0.02 \\ 
  $b_{\unit,4}$ & 1.0 & 11.0 & -- & 0.02 & 0.02 \\ 
  $b_{\unit,5}$ & -2 & 8.00 & -- & 0.02 & 0.02 \\ 
  $b_{\unit,6}$ & -2 & 8.00 & -- & 0.02 & 0.02 \\ 
  $\sigma_{\unit,\mathrm{env}}$ & 0.0 & 1.00 & log & 0.02 & 0.02 \\ 
  $\tilde S^O_{\unit,0}$ & 0.0 & 1.00 & logit & 0.10 & 0.10 \\ 
  $\tilde I^O_{\unit,0}\times10^4$ & 0.0 & 4.00 & logit & 0.20 & 0.20 \\ 
   \bottomrule

\end{tabular}
\end{center}
\caption{Starting values, parameter transformations and perturbation specifications for applying PIF to the polio model.
The first two columns give the lower and upper bounds of a hyper-rectangle sampled uniformly to generate a value used to initialize all particles $\Theta^0_{1:J}$ for each independent PIF replicate.
For the joint maximization, the perturbation sequence used was $\sigma_m=\sigma_0 0.5^{m/50}$.
For the marginal maximization, $\sigma_m=\sigma_{\unit,0}0.25^{m/50}$ was used instead.
Some of these random perturbations were carried out as Gaussian random walks after applying a transformation
to ensure that non-negativity and unit-interval constraints were met, and these transformations are given in the third column.
}\label{tab:polio}
\end{table}

\begin{table}[H]
\begin{center}
\begin{tabular}{ccccc}
\toprule
 & lower bound & upper bound & transformation & $\sigma_0$ \\ 
  \midrule
$\mu_X$ & 0.80 & 3.00 & log & 0.01 \\ 
  $\sigma_X$ & 1.40 & 5.00 & log & 0.01 \\ 
  $\mu_D$ & 1.80 & 7.00 & log & 0.01 \\ 
  $\sigma_D$ & 2.00 & 8.50 & log & 0.01 \\ 
  $\alpha$ & 0.70 & 0.99 & logit & 0.01 \\ 
   \bottomrule

\end{tabular}
\end{center}
\caption{
Starting values, parameter transformations and perturbation specifications for applying PIF to the contacts model.
The first two columns give the lower and upper bounds of a hyper-rectangle sampled uniformly to generate a value used to initialize all particles $\Theta^0_{1:J}$ for each independent PIF replicate.
For the joint maximization, the perturbation sequence used was $\sigma_m=\sigma_0 0.5^{m/50}$.
These random perturbations were carried out as Gaussian random walks after applying the transformations in the third column to ensure that non-negativity and unit-interval constraints were met.
Marginal maximization was not applicable in this example since all parameters were shared.
}\label{tab:contacts}
\end{table}
%% end{add from supp.tex}

\end{document}